\documentclass[journal]{IEEEtran}
\usepackage{xcolor}
\usepackage{amsmath,amsfonts}
\usepackage{amsmath}
\usepackage{array}
\usepackage{textcomp}
\usepackage{url}
\usepackage{verbatim}
\usepackage{graphicx}
\usepackage{cite}
\usepackage{tikz}
\usepackage{amsthm}
\usepackage{subcaption}
\usepackage[normalem]{ulem}
\usepackage{enumitem}
\usepackage{bbm}
\usepackage{amssymb}
	
\newtheorem{remark}{Remark}

\usepackage{mathtools}
\newtheorem{proposition}{Proposition}
\usepackage{breqn}
\usepackage[multiple]{footmisc}

\newtheorem{theorem}{Theorem}
\newtheorem{lemma}{Lemma}

\allowdisplaybreaks

\usepackage[linesnumbered,ruled,vlined]{algorithm2e}

\SetCommentSty{mycommfont}
\SetKwInput{KwInput}{Input}  
\SetKwInput{KwOutput}{Output}       	
\usepackage{hyperref}
\hypersetup{
	colorlinks   = true, 
	urlcolor     = blue, 
	linkcolor    = blue, 
	citecolor   = {blue} 
}

\IEEEoverridecommandlockouts
\begin{document}	
\title{Sub-Sampling for Positioning Privacy in ISAC: Deception by Aliasing via Sparse Arrays and Pilots}
\author{L. Yashvanth, \emph{Member, IEEE}, Christos Masouros, \emph{Fellow, IEEE}, Suraj Srivastava, \emph{Member, IEEE}, Aditya K. Jagannatham, \emph{Senior Member, IEEE}, and Lajos Hanzo, \emph{Life Fellow, IEEE}
	\thanks{L. Yashvanth, and C. Masouros are with the Dept. of Electronic and Electrical Engineering, University College London, WC1E 6BT London, U.K. E-mails: l.yashvanth@ucl.ac.uk, c.masouros@ucl.ac.uk.}
	\thanks{S. Srivastava is with the Dept. of Electrical Engineering, Indian Institute of Technology Jodhpur, Rajasthan 342030, India. E-mail: surajsri@iitj.ac.in.}
	\thanks{A. K. Jagannatham is with the Dept. of Electrical Engineering, Indian Institute of Technology Kanpur, 208016, India. E-mail: adityaj@iitk.ac.in.}
	\thanks{L. Hanzo is with the School of Electronics and Computer Science, University of Southampton, SO17 1BJ, U.K. E-mail: lh@ecs.soton.ac.uk.}
	\thanks{This work will be presented in part at the IEEE  International Workshop Signal Process. and Artif. Intell.  Wireless Commun. (SPAWC), 2026, Greece.}}
\maketitle
\begin{abstract}
Integrated sensing and communications (ISAC) enables simultaneous communication and sensing using shared spectrum and hardware resources in wireless systems. However, securing the sensing functionality against unauthorized receivers remains a fundamental challenge. In this paper, we propose a sub-sampling based sensing-privacy framework for communication-centric (CC)-ISAC systems that jointly exploits sparse arrays and sparse pilot allocations to induce controlled aliasing in the spatial and frequency domains, respectively. By interpreting antenna arrays and pilot subcarriers as spatial and frequency sampling mechanisms, respectively, we show that spatial-frequency undersampling  naturally distorts the range-angle multiple-input multiple-output (MIMO) ambiguity function (AF) observed by an unauthorized receiver. To this end, we first derive a closed-form expression for the range-angle MIMO-AF, and subsequently characterize the ghost targets that arise due to spatial and frequency-domain aliasing. Next, we establish a sufficient condition under which these ambiguities jointly translate into positioning ambiguity and show that, for sufficiently large spatial and frequency sub-sampling factors, an unauthorized receiver inevitably positions a target at incorrect ghost positions. Finally, we show that the proposed sub-sampling framework preserves the native legitimate ISAC performance without introducing additional trade-offs. Numerical results verify the analysis and show that sparse arrays and sparse pilots naturally enable sensing and positioning privacy through deception by aliasing.
\end{abstract}
\vspace{-0.2cm}
\section{Introduction}
Integrated sensing and communications (ISAC) has emerged as a key use case for next-generation wireless systems~\cite{Fan_JSAC_2022,Andrew_JSTSP_2021}. The central idea of ISAC is to utilize common hardware, spectrum, and waveform resources to simultaneously communicate with a user equipment (UE) and sense the surrounding environment. While this joint operation improves spectral and hardware efficiency, it also introduces several new challenges due to the different requirements of communication and sensing functionalities~\cite{Kai_JSAC_2022}. In this paper, we address a specific challenge, namely sensing privacy, where accurate sensing should be achieved only by an authorized receiver, referred to as Alice, while preventing unauthorized sensing by an illegitimate receiver, referred to as Eve. In particular, we conceive a novel sub-sampling-based ISAC signaling framework that jointly exploits sparse arrays and sparse pilot allocations to induce controlled aliasing, thereby creating sensing ambiguity at an unauthorized receiver.
\vspace{-0.2cm}
\subsection{Challenges and Motivation}

A popular instance of ISAC is constituted by dual-function radar-communication (DFRC) systems, wherein a single base station (BS) conveys data to a UE, while simultaneously sensing targets in the surrounding environment using a common waveform~\cite{Fan_JSAC_2022}. However, the shared use of communication and sensing signals also introduces several security vulnerabilities~\cite{Onur_JSAIT_2023}. For example, $a)$ a target may intercept and decode information intended for a communication UE, $b)$ a communication UE may illegitimately infer target parameters that are meant to remain confidential, or $c)$ a more general external agent, referred to as Eve, may perform unauthorized sensing and/or eavesdrop on communication transmissions. In this paper, we focus on the latter case, where an external and silent Eve attempts to illegitimately sense the target parameters exploiting the illumination by the legitimate infrastructure. Addressing this security issue is critical, because unauthorized sensing  enables adversaries to track, identify, or localize targets without permission, which poses significant privacy and security threats to future wireless communications.

In general, securing the sensing functionality of an ISAC system is challenging for several reasons. Firstly, Eve typically remains passive and exploits signals transmitted by the legitimate ISAC BS without revealing its presence. This renders conventional countermeasures such as jamming or spoofing difficult to employ. Secondly, Eve may be located anywhere within the environment; thus it requires solutions for sensing-privacy that remain effective without assuming any prior knowledge of the location or channel characteristics of Eve. Thirdly, since sensing targets/receivers do not actively participate in the communication of data, conventional physical-layer security techniques developed for protecting transmitted data cannot be directly applied to prevent unauthorized sensing. Finally, instead of protecting communications from Eve, sensing privacy seeks to protect target related information. Consequently, encryption does not apply, as Eve may still infer target parameters from the received sensing signatures. 
Therefore, effective mechanisms must be developed for sensing-privacy before ISAC can be deployed in future wireless systems. Such solutions should impair unauthorized sensing via appropriate waveform and/or sensing-process design, while preserving the native ISAC performance at legitimate entities without introducing additional trade-offs. Striking this balance is a fundamental challenge in the design of secure ISAC systems.

\subsection{Related Work and Proposed Idea}
Several design and implementation aspects of ISAC systems  have recently received significant attention~\cite{Fan_TSP_2022,Zhen_TSP_2024,Fan_TIT_2025,Zhang_JSTSP_2025,Peishi_TWC_2025,Kaitao_TWC_2025,Anubhab_TWC_2026,Zhang_Networks_2024,Shihang_TSP_2024,Xiang_JSAC_2022}. For example, transmit precoding was conceived for ISAC in~\cite{Fan_TSP_2022}, constellation design was investigated in~\cite{Zhen_TSP_2024}, and optimal modulation strategies were analyzed in~\cite{Fan_TIT_2025}. Symbol-level precoding techniques were proposed in~\cite{Zhang_JSTSP_2025} and~\cite{Peishi_TWC_2025} for improving the energy efficiency and suppressing waveform sidelobes, respectively. Furthermore, distributed ISAC architectures and the impact of scheduling strategies were investigated in~\cite{Kaitao_TWC_2025} and~\cite{Anubhab_TWC_2026}, respectively. Finally, experimental demonstrations of communication-centric ISAC were reported in~\cite{Zhang_Networks_2024}, while random data-payload based sensing and covariance-constrained precoding designs were studied in~\cite{Shihang_TSP_2024} and~\cite{Xiang_JSAC_2022}, respectively. 

Despite these developments, securing the sensing functionality of ISAC systems remains a largely open challenge. Early approaches in~\cite{Nanchi_TWC_2021,Nanchi_TWC_2022} focused on preventing the interception of communication signals by a sensing target upon exploiting artificial noise and interference, respectively. However, as discussed earlier, protecting target information from an external and passive sensing receiver is considerably more challenging. In this context, the authors of~\cite{Jiaqi_TVT_2024} proposed the use of artificial noise to impair unauthorized sensing. However, such approaches typically require knowledge of Eve's channel characteristics, which may not be readily available in practice. More recently,~\cite{Kawon_TWC_2026} proposed to deliberately distort the ambiguity function (AF) of the transmit waveform, while~\cite{chen2025sensing} advocated illuminating scatterers for creating clutter to confuse the sensing signatures observed by Eve. Although effective in degrading the unauthorized sensing performance, such ideas may also compromise the legitimate ISAC performance.

More generally, the above techniques tend to achieve privacy by deliberately modifying the transmit waveform or sensing environment, which often comes at the expense of eroding the legitimate ISAC performance. To circumvent this, we adopt a fundamentally different approach. Specifically, we show that sensing privacy may be naturally gleaned from the sub-sampling of space-time-frequency signals in the spatial and frequency domains through sparse arrays and sparse pilot allocations, respectively. To unveil our idea, we interpret the reception of a signal using a multiple-antenna array as a spatial sampling process, where the signal observed at each antenna corresponds to a spatial sample of an underlying space-time-frequency signal. Then, the inter-antenna spacing determines the spatial sampling rate, and half-carrier wavelength antenna spacing corresponds to the Nyquist spatial-sampling rate~\cite{johnson1992array}. So, antenna spacings larger than half a wavelength result in spatial undersampling and hence introduce spatial aliasing. In this paper, we consider antenna spacings that are integer multiples of half a wavelength. Equivalently, this can be viewed as activating only a smaller subset of antennas from a Nyquist-sampled array, leading to what we refer to as a sparse array~\cite{Moeness_ProcIEEE_2016,johnson1992array}. Such spatial undersampling introduces grating lobes and spatial aliasing, which create additional virtual spatial frequencies and ambiguity in the observed direction of arrival (DoA.) Similar observations apply in the frequency domain (FD) of an orthogonal frequency division multiplexing (OFDM) system. Specifically, we periodically place sparse sensing-pilot subcarriers (SCs) across the OFDM grid, and this can be interpreted as FD sub-sampling. Then, this introduces FD aliases, analogous to the spatial domain (SD) aliases introduced by sparse arrays. By jointly exploiting spatial and FD aliases, we induce positioning ambiguity at an Eve for securing ISAC. Note that, owing to its cooperation with the ISAC BS, Alice can utilize both the pilots and data-payload for sensing, whereas Eve is restricted to sensing using only the pilot symbols. The key contributions of our paper are:

\begin{enumerate}[leftmargin=*]
\item For the proposed ISAC signalling, we derive a closed-form expression of the range-angle multiple-input multiple-output (MIMO) AF at Eve and show that it has a separable representation in the range-angle domains. (See Lemma~\ref{lem_AF_Eve_expression}.)
\item We prove that the transmit spatial MIMO-AF exhibits multiple peaks, when the BS employs a sparse array. Using this, we characterize the ghost targets observed by Eve in the angle domain under a maximum likelihood (ML)-criterion. This shows that sparse arrays promote sensing privacy in the SD. (See Lemma~\ref{lem_ambiguity_function} and Theorem~\ref{thm_ghost_target_location}.)
\item Similarly, we show that the range AF exhibits multiple local maxima when sparse sensing pilots are periodically placed over the OFDM SCs. Using this, we characterize the ghost targets observed by Eve in the range domain under the ML-criterion. (See Lemma~\ref{lem_range_AF_peaks} and Theorem~\ref{thm_range_ambiguity}.)
\item Using the expressions for range-angle ghost targets, we derive a sufficient condition for ensuring that Eve has an overall positioning level ambiguity for a target. We then show that provided the SD and FD sub-sampling factors are sufficiently high, Eve always experiences ghost positioning of a target. (See Theorem~\ref{thm_localization_ambiguity} and Proposition~\ref{prop_existance_of_position_error}.) 
\item When the UE performs linear minimum mean square estimation (LMMSE) of channels on data SCs based on least squares (LS) channel estimates obtained on pilot SCs, we show that the sub-sampling process preserves the native performance of the legitimate ISAC system without imposing any additional trade-offs. (See Theorem~\ref{prop_legitimate_performance}).
\end{enumerate}
Overall, we demonstrate how the sub-sampling of signals in the SD-FD can naturally facilitate sensing privacy in ISAC applications, thanks to the intentional aliasing introduced in the spatial-range domains, which enables sensing mis-information at an Eve. Importantly, this approach performs well regardless of where Eve is located, while preserving the performance of legitimate ISAC nodes. Our numerical findings reveal that with $16\%$ spatial- and $1.5\%$ pilot-sparsity levels, we can create a positioning privacy gap of $80$ meters between Alice and Eve.

\emph{Notations:} The mathematical notations and symbols carry usual meanings. In particular, $[n] \triangleq \{1,\ldots,n\}$ is the set of first $n$ natural numbers; $\biguplus$ is the disjoint set-union operator; $\text{det}(\mathbf{X})$ is the determinant of matrix $\mathbf{X}$; $\textrm{blkdiag}(\mathbf{X}_1,\ldots,\mathbf{X}_n)$ is a block-diagonal matrix with $\mathbf{X}_1,\ldots,\mathbf{X}_n$ on its diagonal; $\otimes$ and $\odot$ denote Kronecker and Hadamard products, respectively.

\section{System Model}\label{sec_sys_model}
We consider a communication-centric ISAC system in which an $N_t$-antenna BS transmits data in the downlink (DL) to $K_c$ UEs, while simultaneously sensing $K_s$ radar targets\footnote{For simplicity, we assume stationary targets and neglect Doppler effects.} using an OFDM waveform. The $K_c$ UEs, each equipped with $N_c$ receive antennas, are served in a time-division multiple-access (TDMA) manner. The system uses a total of $N_s$ SCs at a carrier frequency of $f_c$ spanning a bandwidth (BW) of $B$. Sensing is accomplished using a legitimate $N_r$ antenna receiver (called Alice), collocated with the BS under the monostatic mode. At the same time, an illegitimate receiver (called Eve), equipped with $N_e$ antennas, passively senses these targets using the signals transmitted by the ISAC BS, as shown in Fig.~\ref{fig_system_model}. Alice and Eve perform sensing within the unambiguous operating range of radar.
\vspace{-0.2cm}

\subsection{Sparse Antenna Array and ISAC Transmit Signal Model}
The BS array is configured as a uniform sparse array (USA) with inter-antenna spacing equal to $\eta_s\lambda_0/2$, where $\lambda_0$ is the carrier wavelength, and $\eta_s\geq1$, $\eta_s \in \mathbb{Z}_+$ is the spatial sub-sampling factor at the BS array. Note that $\eta_s=1$ corresponds to the Nyquist sampling rate with $\lambda_0/2$-spacing in the SD; while $\eta_s>1$ creates sparse arrays. 
The communication UE, Alice and Eve use a compact ULA, each with inter-antenna spacing equal to $\lambda_0/2$.  

Let $\mathbf{x}(t)$ be the signal transmitted from the ISAC BS at time $t$ in the complex base-band domain. Then we can write
\begin{equation}
\mathbf{x}(t) = \frac{1}{\sqrt{N_s}}\sum\nolimits_{n=0}^{N_s-1}\mathbf{x}_n e^{j2\pi n \Delta ft} \mathrm{rect}\left(\frac{t}{T}\right),
\end{equation}
where $\mathbf{x}_n \in \mathbb{C}^{N_t}$ is the signal transmitted from ISAC BS array on the $n$th SC, $\Delta f = B/N_s$ is the BW of a SC, $T = (1/\Delta f) + T_{\textrm{CP}}$ is the total duration of an OFDM symbol including the cyclic prefix (CP) that spans over $T_{\textrm{CP}}$ time units, and $\mathrm{rect}(\cdot)$ is the rectangular transmit-pulse used for communications.
\begin{figure}
	\vspace{-0.2cm}
\centering
\includegraphics[width=0.9\linewidth]{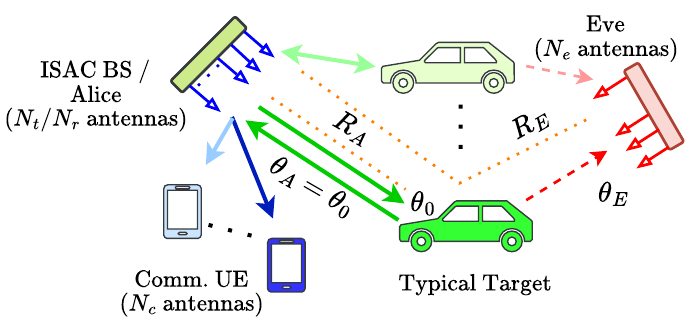}
\caption{ISAC with legitimate Rx. Alice and illegitimate Rx. Eve.}
\label{fig_system_model}
\end{figure}
\begin{figure}
	\vspace{-0.2cm}
\centering
\includegraphics[width=0.92\linewidth]{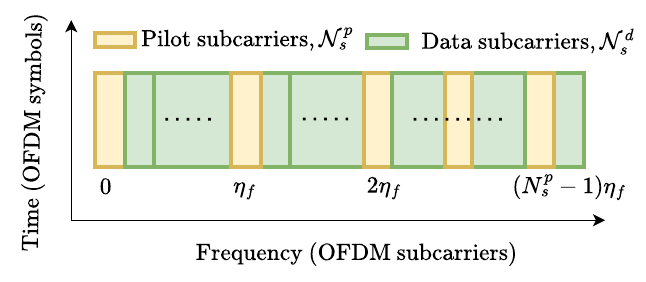}
\caption{Subcarrier arrangement for pilot and data in OFDM.}
\label{fig_OFDM_structure}
\vspace{-0.3cm}
\end{figure}
\vspace{-0.2cm}

\subsection{Sub-carrier Arrangement for Sparse Pilots in OFDM}
In line with the $5$G standards, we assume that a few of the SCs are reserved for pilot transmission, which can be used for channel estimation and sensing, while the remaining SCs carry data payload for communication with a UE. Let $N_s^p$ and $N_s^d$ denote the number of SCs allotted for pilot and data, respectively; then, we have $N_s^p + N_s^d = N_s$. Furthermore, let $\mathcal{N}_s \triangleq \{0,1,\ldots,N_s-1\}$, $\mathcal{N}_s^p$, and  $\mathcal{N}_s^d$ denote the index set of all SCs, and of the SCs allotted for pilot and data, respectively, such that $\mathcal{N}^p_s \biguplus \mathcal{N}_s^d = \mathcal{N}_s$. In this work, we periodically place the pilots across the SCs. More specifically
\vspace{-0.1cm}
\begin{equation}
	\mathcal{N}_s^p \triangleq \{0,\eta_f,2\eta_f,\ldots,(N_s^p-1)\eta_f\},
\end{equation}
where $\eta_f \!\geq\! 1, \eta_f\!\in\!\mathbb{Z}_+$ is the frequency sub-sampling factor, and governs the number of pilot SCs as $N_s^p\!=\!\left\lfloor\left(N_s\!-\!1\right)\big/\eta_f\right\rfloor\!+\!1$. Increasing $\eta_f$ leads to more widely spaced pilot SCs, resulting in sparser pilot SCs. In Fig.~\ref{fig_OFDM_structure}, we pictorially depict the SC-arrangement for pilot and data in the OFDM grid.\footnote{The proposed framework can also be extended to Doppler domain by incorporating sparse sensing pilots along the temporal dimension. However, it is beyond the scope of this paper and is left for future work.}
\vspace{-0.2cm}
\subsection{Sensing System Model}
Under the sensing operation, the echo signals obtained from target reflections are received by Alice and Eve. Assume that the $k$th target ($k=1,\ldots,K_s$) to be sensed is based at an angle $\theta^k_0$ with respect to the BS antenna array axis, and at $\theta^k_A$ and $\theta^k_E$ w.r.t. the receive arrays of Alice and Eve, respectively. Moreover, under monostatic sensing by Alice, we have $\theta^k_A = \theta^k_0$. The sensing channels observed by the antenna arrays are modeled using the array steering response vectors, defined as 
\begin{equation}\label{eq_array_steering_vec}
\mathbf{a}_N(\theta;\eta_s) \triangleq [1,e^{-j\pi\eta_s\sin(\theta)},\ldots,e^{-j\pi(N-1)\eta_s\sin(\theta)}]^T,
\vspace{-0.1cm}
\end{equation}
where $N$ is the number of antenna elements, $\eta_s$ is the array spatial-sampling factor, and $\theta$ is the path angle constituting the channel. For notational brevity, when $
\eta_s=1$, we adopt the re-parameterization of $\tilde{\mathbf{a}}_{N}(\theta
) \!\triangleq \mathbf{a}_N(\theta;1)$. 

At Alice, who can cooperate with the BS, sensing is performed using both pilot and data SCs. By contrast, since Eve is assumed to remain silent and disconnected from the network, she does not have access to the information transmitted on the data SCs, but she can perform target sensing using the pilot SCs. More specifically, exploiting the data SCs would require seamless error-free decoding of random data symbols, which is infeasible due to Eve's lack of knowledge of the BS-Eve channel. Therefore, we assume that Eve knows $\mathbf{x}_n$ only for $n \in \mathcal{N}_s^p$ and passively senses targets using the SCs indexed by $\mathcal{N}_s^p$. We next present the signal models at Alice and Eve.

\subsubsection{Alice's Signal Model}
The target echo received at Alice in the base-band domain at time $t$ can be written as $\mathbf{y}_a(t)=$
\begin{align}
 & \sum\nolimits_{k=1}^{K_s}\beta^k_{a}\tilde{\mathbf{a}}_{N_r}(\theta^k_0)\mathbf{a}^H_{N_t}(\theta^k_0;\eta_s)\mathbf{x}(t-\tau^k_A)e^{-j2\pi f_c  \tau^k_A} + \mathbf{n}_{a}(t), \nonumber \\
&= \frac{1}{\sqrt{N_s}}\sum\nolimits_{n=0}^{N_s-1}\sum\nolimits_{k=1}^{K_s}\beta^k_{a}\tilde{\mathbf{a}}_{N_r}(\theta^k_0)\mathbf{a}^H_{N_t}(\theta^k_0;\eta_s) \mathbf{x}_n e^{-j2\pi f_c\tau^k_A} \nonumber\\
& \hspace{1cm}\times e^{j2\pi n \Delta f(t-\tau^k_A)} \mathrm{rect}\left(\frac{t-\tau^k_A}{T}\right) + \mathbf{n}_{a}(t), \label{eq_time_signal_alice}
\end{align}
where $\tau^k_A \geq 0$ is the round-trip-time (RTT) of the transmit signal at Alice after reflection from the $k$th target, and $\mathbf{n}_a(t) \in \mathbb{C}^{N_r}$ is the additive white Gaussian noise at Alice. When the RTT of all the targets is less than the CP duration, by sampling the received signal at a rate of $T_s = p/B$, $p \in \mathbb{Z}_+$ we obtain the $p$th sample of the received signal vector at Alice as
\vspace{-0.1cm}
\begin{multline}\label{eq_rcvd_signal_BB}
\mathbf{y}_{a,p} = \frac{1}{\sqrt{N_s}}\sum\nolimits_{k=1}^{K_s}\Bar{\beta}^k_a \tilde{\mathbf{a}}_{N_r}(\theta^k_0)\mathbf{a}^H_{N_t}(\theta^k_0;\eta_s) \sum\nolimits_{n=0}^{N_s-1} \mathbf{x}_n \\ \times e^{j2\pi \frac{n}{N_s}p} \times e^{-j2\pi n \Delta f \tau^k_A} + \mathbf{n}_{a,p},
\vspace{-0.1cm}
\end{multline}
where $\bar{\beta}^k_a \triangleq \beta^k_a e^{-j2\pi f_c\tau^k_A}$ is the complex channel gain between the ISAC BS and Alice accounting for the $k$th target's radar cross section (RCS) gain and the path-loss, while $\mathbf{n}_{a,p} \in \mathbb{C}^{N_r}$ is the noise vector at Alice. Subsequently, collecting all samples within the OFDM symbol and discarding the samples corresponding to the CP, we obtain the FD representation of the received signal in~\eqref{eq_rcvd_signal_BB}, denoted by $\tilde{\mathbf{Y}}_{a} \in \mathbb{C}^{N_r \times N_s}$ as $ \tilde{\mathbf{Y}}_{a} =$
\begin{equation}\label{eq_rcv_matix_one_OFDM_symbols}
\!\sum_{k=1}^{K_s}\!\bar{\beta}^k_a\tilde{\mathbf{a}}_{N_r}(\theta^k_0)\Big\{\!\Big(\!\mathbf{a}^H_{N_t}(\theta^k_0;\eta_s)\underbrace{\left[\mathbf{x}_0,\ldots,\mathbf{x}_{N_s-1}\right]}_{\triangleq \tilde{\mathbf{X}}}\Big) \odot \mathbf{d}_a^T(\tau^k_A)\Big\} + \tilde{\mathbf{N}}_a,
\end{equation}
where $\mathbf{d}_a(\tau) \in \mathbb{C}^{N_s}$ is the \emph{frequency-steering response vector} of OFDM at Alice corresponding to delay $\tau$, given by
\begin{equation}\label{eq_freq_steering_vec}
\mathbf{d}_a(\tau) \triangleq \left[1,e^{-j2\pi\Delta f \tau},e^{-j4\pi\Delta f \tau},\ldots,e^{-j2\pi(N_s-1)\Delta f \tau}\right]^T,
\end{equation}
and $\tilde{\mathbf{N}}_a$ is the noise matrix at Alice in the FD.
\begin{figure*}[!ht]
	\vspace{-0.8cm}
	\begin{equation}\label{eq_Alice_signal_model}
		\mathbf{Y}_a = \sum_{k=1}^{K_s}\!\bar{\beta}^k_a\tilde{\mathbf{a}}_{N_r}(\theta^k_0)  \Big\{\!\Big(\!\mathbf{a}^H_{N_t}(\theta^k_0;\eta_s)    \underbrace{\left[\mathbf{x}_{0,1},\ldots,\mathbf{x}_{0,L},\mathbf{x}_{1,1},\ldots,\mathbf{x}_{1,L},\ldots,\mathbf{x}_{N_s-1,1},\ldots,\mathbf{x}_{N_s-1,L}\right]}_{\triangleq {\mathbf{X}} \in \mathbb{C}^{N_t\times N_sL}}\Big)\! \odot\! \left(\mathbf{d}_a(\tau^k_A)\otimes\mathbf{1}_L\right)^T\!\Big\}\! +\! \mathbf{N}_a.
	\end{equation}
	\hrule
\end{figure*}
\begin{figure*}[!ht] 
	\vspace{-0.5cm}
	\begin{equation}\label{eq_Eve}
	\!\!\mathbf{Y}_e \!=\!\! \sum_{k=1}^{K_s}\! \bar{\beta}^k_e\tilde{\mathbf{a}}_{N_e}\!(\theta^k_E)\Big\{\!\Big(\!\mathbf{a}^H_{N_t}\!(\theta^k_0;\eta_s)\!\! \underbrace{\left[\mathbf{x}_{0,1},\!\ldots,\!\mathbf{x}_{0,L},\mathbf{x}_{\eta_f,1},\!\ldots,\!\mathbf{x}_{\eta_f,L},\!\ldots,\!\mathbf{x}_{(N_s^p\!-\!1)\eta_f,1},\ldots,\mathbf{x}_{(N_s^p\!-\!1)\eta_f,L}\right]}_{\triangleq {\mathbf{X}_e} \in \mathbb{C}^{N_t\times N_s^pL}}\!\Big) \odot \left(\mathbf{d}_e(\tau^k_E)\!\otimes\!\mathbf{1}_L\!\right)^T\!\!\Big\} \!+\! \mathbf{N}_e,
	\end{equation}
	\hrule
	\vspace{-0.3cm}
\end{figure*}
Note that~\eqref{eq_rcv_matix_one_OFDM_symbols} corresponds to the received space-frequency signal after processing one OFDM symbol. In practice, communication is a continuous process and the OFDM symbols are incessantly transmitted and received. We can therefore exploit this continuum of data, collect multiple measurements of target parameters, and jointly process them to improve the sensing performance. This effectively increases the radar's coherent processing interval (CPI). Let $L \geq N_t$ be the total number of OFDM symbols, and $\mathbf{x}_{n,\ell}$ be the signal vector transmitted from the BS array on $n$th SC of $\ell$th OFDM symbol. Then, the space-time-frequency (STF) domain signal, $ \mathbf{Y}_a \in \mathbb{C}^{N_r \times N_sL}$ at Alice is given in~\eqref{eq_Alice_signal_model}, where $\mathbf{N}_a$ is the total noise at Alice and $\mathbf{X} \in \mathbb{C}^{N_t \times N_sL}$ is the transmit signal in the STF domain. From~\eqref{eq_Alice_signal_model}, the sensing parameters of interest at Alice correspond to the DoAs, $\left\{\theta^k_0\right\}_{k=1}^{K_s}$ and range of the targets: $\left\{R^k_A\right\}_{k=1}^{K_s}$, where $R^k_A$ is related to the RTT as $R^k_A = c\tau^k_A/2$ with $c=3\!\times \!10^8 \ m/s$ being the speed of light. 

\subsubsection{Eve Signal Model}
Recall that Eve senses the target using pilot SCs. Then, similar to~\eqref{eq_Alice_signal_model}, after reflection of the signal from the target, the signal received at Eve, $\mathbf{Y}_e \in \mathbb{C}^{N_e \times N^p_sL}$, can be obtained as in~\eqref{eq_Eve},
where $\mathbf{d}_e(\tau_E) \in \mathbb{C}^{N_s^p}$ is the frequency-steering response vector at Eve, given by
\begin{equation}\label{eq_freq_steering_vec_eve}
\mathbf{d}_e(\tau) = \left[1,e^{-j2\pi\eta_f\Delta f \tau},e^{-j4\pi\eta_f\Delta f \tau},\ldots,e^{-j2\pi(N_s^p-1)\eta_f\Delta f \tau}\right]^T\!\!,
\end{equation}
and $\tau^k_E \geq 0$ is the time-of-flight (TOF) of the signal from the ISAC BS to Eve after reflection from the $k$th target. The TOF captures the range information of the target under bistatic mode. Then, Eve estimates the directions of departure (DoD) $\left\{\theta^k_0\right\}_{k=1}^{K_s}$, DoAs $\left\{\theta^k_E\right\}_{k=1}^{K_s}$, and TOFs $\left\{\tau^k_E\right\}_{k=1}^{K_s}$. 

Finally, we recognize that under the OFDM framework considered, the maximum un-ambiguous range of the sensing system equals $c/\Delta f$. Accordingly, in this paper, we restrict our focus to cases where $\tau^k_A, \tau^k_E \in (0, 1/\Delta f)$. 
\vspace{-0.2cm}
\subsection{Communication System Model and Signal Design}\label{sec:comm_sys_model}
Since TDMA assigns all SCs to a single UE at any given time, without loss of generality, we present the analysis and results for a single UE in the rest of the paper. Accordingly, let $\mathbf{H}_{c,n}\in\mathbb{C}^{N_c\times N_t}$ be the BS-UE channel on the $n$-th SC. For simplicity, we assume that $\mathbf{H}_{c,n}$ has full rank. 
Let $\mathbf{X}_n \triangleq \left[\mathbf{X}\right]_{:,(n-1)L+1:nL} \in \mathbb{C}^{N_t \times L}$ collect the transmit signal vectors on the $n$th SC. 
Then, similar to~\eqref{eq_Alice_signal_model}, the signal received at the UE on the STF grid is given by
\begin{equation}\label{eq_comm_rcvd_signal}
\mathbf{Y}_{c} = \left[\mathbf{H}_{c,1},\ldots,\mathbf{H}_{c,N_s}\right]\textrm{blkdiag}\left(\mathbf{X}_1,\ldots,\mathbf{X}_{N_s}\right) + \mathbf{N}_c,
\end{equation}
where $\mathbf{N}_c \in \mathbb{C}^{N_c \times N_sL}$ is the additive Gaussian noise at the UE. Let $\mathbf{n}_{c,n,\ell}$ be the UE-noise vector on the $n$th SC of $\ell$th OFDM symbol. We use the model $\mathbf{n}_{c,n,\ell} \sim \mathcal{CN}(0,\sigma_c^2\mathbf{I}_{N_c})$.

On the data SCs, the transmit signal, $\mathbf{x}_n$ is designed as
\begin{equation}
\mathbf{x}_n = \mathbf{W}_n\mathbf{s}_n, \ n \in \mathcal{N}_s^d,
\end{equation}
where $\mathbf{W}_n \in \mathbb{C}^{N_t\times N_t}$ is the transmit precoding matrix (TPM) and $\mathbf{s}_n \in \mathbb{C}^{N_t}$ collects the information symbols transmitted over $N_t$ antennas on the $n$th SC, which are drawn from i.i.d. Gaussian distribution with zero mean and unit-variance. Note that Gaussian signaling is capacity-optimal in systems communicating over Gaussian channels. Furthermore, the information-theoretic rate-optimal transmit signal is the one that maximizes~\cite{telatar1999capacity}
\begin{equation}\label{eq_througput}
\mathcal{R} \triangleq \frac{1}{N_s}\sum\nolimits_{n \in \mathcal{N}^d_s}\log_2\text{det}\left(\mathbf{I}_{N_c}+ \gamma_{c,n}\mathbf{H}_{c,n}\mathbf{\Sigma}_{X,n}\mathbf{H}_{c,n}^{H}\right),
\end{equation}
where $\mathbf{\Sigma}_{X,n} \in \mathbb{C}^{N_t \times N_t}$ denotes the transmit signal covariance matrix on the $n$th SC, and $\gamma_{c,n}$ is the transmit signal-to-noise ratio (SNR): $P_n/\sigma_c^2$ with $P_n$ being the total transmit power available for $n$th SC. Now, for analytical tractability, we assume that the BS allots equal power to all available spatial eigen modes and SCs. This can be justified because $a)$ at high SNR, equal power allocation is the rate-optimal strategy, $b)$ in practice, it can be readily implemented without much loss in performance in i.i.d. scenarios, and $c)$ equally allotting power across all spatial eigen modes is also information-theoretically rate-optimal, when the BS only has statistical channel state information (CSI) of the communication UE~\cite{telatar1999capacity}. As a result, we have $P_n = P/N_s$, where $P$ is the total available power at ISAC BS, and the optimal precoding matrix for the $n$th SC is given by $\mathbf{W}_n =\frac{1}{\sqrt{N_t}}\mathbf{V}_{\mathbf{H}_{c,n}}$, where $\mathbf{V}_{\mathbf{H}_{c,n}} \in \mathbb{C}^{N_t \times N_t}$ is the right singular matrix of $\mathbf{H}_{c,n}$~\cite{telatar1999capacity}. 
Then, we have $\boldsymbol{\Sigma}_{X,n} =$
\begin{equation}    
\frac{1}{L}\mathbb{E}\left[\mathbf{X}_n\mathbf{X}_n^H\right] = \frac{1}{L}\mathbf{W}_n\mathbb{E}\left[\sum_{\ell=1}^L\mathbf{s}_{n,\ell}\mathbf{s}_{n,\ell}^H\right]\mathbf{W}_n^H = \frac{\mathbf{V}_{\mathbf{H}_{c,n}}\mathbf{V}_{\mathbf{H}_{c,n}}^H}{N_t},
\end{equation}
where we used $\mathbb{E}[{s}_{n,\ell}\mathbf{s}_{n,\ell}^H] = \mathbf{I}_{N_t}$ for all $n \in \mathcal{N}_s,\ell\in [L]$. Then noting that $\mathbf{V}_{\mathbf{H}_{c,n}}$ is a unitary matrix, it follows that $\mathbf{\Sigma}_{X,n} = (1/N_t)\mathbf{I}_{N_t}$.

Similarly, for pilot SCs on the $\ell$th OFDM symbol, we have 
\begin{equation}\label{eq_pilot_signal}
\mathbf{x}_n = [\boldsymbol{\Phi}]_{:,\ell\!\!\!\mod\!N_t}, \ n\in\mathcal{N}_s^p,
\end{equation}
where $[\boldsymbol{\Phi}]_{:,q} \in \mathbb{C}^{N_t}$ is a deterministic sensing signal chosen as the $q$th column of a suitable $\mathbb{C}^{N_t\times N_t}$-dimensional codebook that has favorable properties for sensing/estimation tasks, e.g., a Hadamard or discrete Fourier transform (DFT) matrix. In that case, it follows that $\boldsymbol{\Phi}\boldsymbol{\Phi}^H = \mathbf{I}_{N_t}$. 
\section{Spatial-Range Sub-Sampling Promotes Sensing Privacy}

In this section, we analytically show how sparse arrays combined with periodically placed sparse pilots over the OFDM SCs can naturally create sensing ambiguities at Eve along the range-angle domain.

\subsection{The Range-Angle MIMO Ambiguity Function at Eve}
The AF of an ISAC waveform characterizes the resolution with which the target parameters can be estimated, and it corresponds to the deterministic output of a receiver matched filter (MF)~\cite{levanon2004radar}. Assume that Eve attempts to sense the targets characterized by the angle-delay tuple $(\boldsymbol{\theta}_0,\boldsymbol{\theta}_E,\boldsymbol{\tau}_E)$ using a MF steered toward a reference angle-delay tuple $(\theta_0',\theta_E',\tau_E')$, where $\boldsymbol{\theta}_0 \triangleq [\theta_0^1,\ldots,\theta_0^{K_s}]^T$, $\boldsymbol{\theta}_E \triangleq [\theta_E^1,\ldots,\theta_E^{K_s}]^T$, and $\boldsymbol{\tau}_E \triangleq [\tau_E^1,\ldots,\tau_E^{K_s}]^T$. Then, we have the following result. 
\begin{lemma}\label{lem_AF_Eve_expression}
For the ISAC system described in Sec.~\ref{sec_sys_model}, the joint range-angle MIMO AF at Eve corresponding to targets having the parameters $(\boldsymbol{\theta}_0,\boldsymbol{\theta}_E,\boldsymbol{\tau}_E)$ and evaluated at the reference parameters $(\theta_0',\theta_E',\tau_E')$ admits the following representation
\vspace{-0.1cm}
\begin{multline}\label{eq_eve_AF_seperable_form}
	\psi_E(\boldsymbol{\theta}_0,\theta_0',\boldsymbol{\theta}_E,\theta_E',\boldsymbol{\tau}_E,\tau_E') \\ = \sum_{k=1}^{K_s} \psi_{EA}(\theta_E',\theta^k_E)\psi_{BS}(\theta_0',\theta^k_0,\eta_s)\psi_{ER}(\tau_E',\tau^k_E,\eta_f),
\end{multline}
where we have
\begin{align}
\hspace{-0.5cm}\psi_{EA}(\theta_E',\theta^k_E) &=\left|{\sin\left(\frac{\pi N_e}{2}\Delta^k_{\theta_E',\theta_E}\right)}\Big/{N_e\sin\left(\frac{\pi}{2}\Delta^k_{\theta_E',\theta_E}\right)}\right|,\label{eq_AMB_Eveangle}\\
	\hspace{-0.4cm}\psi_{BS}(\theta_0',\theta^k_0,\eta_s) &= \left|{\sin\!\left(\!\frac{\pi N_t}{2}\eta_s\Delta^k_{\theta_0',\theta_0}\!\right)}\!\Big/\!{N_t\sin\!\left(\frac{\pi}{2}\eta_s\Delta^k_{\theta_0',\theta_0}\right)}\right|\!,\label{eq_AMB_BS}\\
	\hspace{-0.1cm}\psi_{ER}(\tau_E',\tau^k_E,\eta_f) &=\left|\frac{\sin\left({\pi N_s^p \eta_f\Delta f}\Delta^k_{\tau_E',\tau_E}\right)}{N^p_s\sin\left({\pi \eta_f\Delta f}\Delta^k_{\tau_E',\tau_E}\right)}\right|. \label{eq_AMB_Everange}
\end{align}
Here, $\Delta^k_{\theta_0',\theta_0}
\triangleq
\sin(\theta_0')-\sin(\theta^k_0),
\Delta^k_{\theta_E',\theta_E}
\triangleq
\sin(\theta_E')-\sin(\theta^k_E)$,
denote the DoD and DoA mismatches, respectively, with
$\Delta^k_{\theta_0',\theta_0},
\Delta^k_{\theta_E',\theta_E}
\in
\Omega_s \triangleq (-2,2)$,
while
$\Delta^k_{\tau_E',\tau_E}
\triangleq
\tau_E'-\tau^k_E
\in
\Omega_f
\triangleq
\left(-\frac{1}{\Delta f},\frac{1}{\Delta f}\right)$ 
denotes the TOF/delay mismatch.
\end{lemma}
\begin{proof}
	See Appendix~\ref{app_Eve_AF}.
\end{proof}

Observe from Lemma~\ref{lem_AF_Eve_expression} that for any given target, the range-angle MIMO-AF is separable across the receive-array, transmit-array, and delay dimensions. Consequently, the estimation of angle and range parameters can be analyzed independently, allowing us to study the sensing privacy of each parameter separately. Motivated by this, we next investigate how sparse arrays and sparse pilot allocations naturally induce ambiguities in the spatial and range domains, respectively. Furthermore, observe from~\eqref{eq_eve_AF_seperable_form} that the overall AF is given by the superposition of the AFs corresponding to individual targets. Therefore, all subsequent results are presented for a given target and apply identically to more than one target. Accordingly, target index $k$ is omitted for notational brevity.
\subsection{Spatial-Sensing Privacy Against Eve}\label{sec_security_Eve}
For ideal sensing of target angles at Eve, AFs must satisfy
\begin{equation}\label{eq_AF_reqd}
\psi_{BS}(\theta_0',\theta_0,\eta_s) = \phi(\theta_0'-\theta_0), \quad
\psi_{EA}(\theta_E',\theta_E) = \phi(\theta_E'-\theta_E),
\end{equation}
where $\phi(x)$ attains its maximum at $x=0$ and rapidly decays to zero for $x\neq 0$ (e.g., $\phi(x)=\delta(x)$, the Dirac delta function). If either of the conditions in~\eqref{eq_AF_reqd} is violated, accurate sensing becomes infeasible.
Next, we show that the transmit spatial MIMO-AF, $\psi_{BS}(\theta_0',\theta_0,\eta_s)$ does not satisfy~\eqref{eq_AF_reqd} when $\eta_s>1$. 

\begin{lemma}\label{lem_ambiguity_function}
Let $\Delta_s \triangleq \theta_0'-\theta_0 \in \Omega_s$ be the  DoD-spatial difference. Then, the transmit spatial MIMO-AF at Eve using an $N_t$-element uniform sparse array with spatial-sampling factor $\eta_s$, attains local maxima at 
\begin{equation}\label{eq_peak_AF_values}
	\Delta_{s,n} = 2n/\eta_s, \quad n=0, \pm 1, \pm 2, \ldots, \lfloor \eta_s \rfloor.
\end{equation}
\end{lemma}
\begin{proof}
	See Appendix~\ref{app_SP_AF_Peak}.
\end{proof}
\vspace{-0.1cm}
Next, we exploit the above result and show how sparse arrays can promote spatial-sensing privacy against Eve.

\begin{theorem}\label{thm_ghost_target_location}
Consider a uniform sparse array based CC-ISAC BS equipped with $N_t$ antennas and spatial-sampling factor $\eta_s$ as described in Sec.~\ref{sec_sys_model}. If an illegitimate receiver Eve attempts to estimate the DoD parameter $\theta_0$ of the target using the maximum-likelihood (ML) criterion, then it also observes \emph{ghost target DoD estimates} given by
\vspace{-0.1cm} 
\begin{equation}\label{eq_ghost_target_location}
	\hat{\theta}_{0,n} =
	\sin_{(p)}^{-1}\!\left(\sin(\theta_0)+\frac{2n}{\eta_s}\right),
	\qquad
	n \in \mathcal{N}_g(\theta_0),
	\vspace{-0.1cm} 
\end{equation}
where $\sin_{(p)}^{-1}(\cdot)$ denotes the inverse sine function within its principal domain, $[-1,1]$ and $\mathcal{N}_g(\theta) \triangleq$
\vspace{-0.1cm} 
\begin{equation}\label{eq_gh_location_ind}
	\left\{ n \in \mathbb{Z}\setminus\{0\}
	\Big|
	\left\lceil
	-\frac{\eta_s}{2}(1+\sin\theta_0)
	\right\rceil \leq n \leq \left\lfloor
	\frac{\eta_s}{2}(1-\sin\theta_0)
	\right\rfloor
	\right\}.
\end{equation}
\vspace{-0.5cm}
\end{theorem}
\begin{proof}
	See Appendix~\ref{app_SP_Ghost}.
\end{proof}
Observe from Theorem~\ref{thm_ghost_target_location} that for any true target parameter $\theta_0$ and location of Eve, we have $|\mathcal{N}_g(\theta_0)| > 1$, i.e., Eve always witnesses ghost targets. 
Hence, even if Eve may be able to estimate the target DoA, it cannot precisely estimate the target DoD, thereby securing the sensing operation in the SD.
\subsection{Range-Sensing Privacy Against Eve}\label{sec_security_range_Eve}
Similar to Sec.~\ref{sec_security_Eve}, we next deduce how sparse pilots promote range-domain ambiguity at Eve. 

\begin{lemma}\label{lem_range_AF_peaks}
Let $\Delta_\tau \triangleq \tau_E'-\tau_E  \in \Omega_f$ be the delay difference. Then, the range-AF at Eve using an $N_s$-SC OFDM waveform with frequency-sampling factor $\eta_f$, attains local maxima at 
\begin{equation}\label{eq_Eve_range_AF_maxima}
	\Delta_{\tau,m} = \frac{m}{\eta_f \Delta f}, \quad m = 0, \pm 1, \pm 2, \ldots, \pm \left \lfloor \eta_f \right \rfloor.
\end{equation}
\end{lemma}
\begin{proof}
	Similar to proof of Lemma~\ref{lem_ambiguity_function}. Skipped for brevity. 
\end{proof}
Now, similar to Theorem~\ref{thm_ghost_target_location}, we deduce the range-domain ghost targets at Eve in the following result. 
\begin{theorem}\label{thm_range_ambiguity}
Consider the CC-ISAC system with pilot SCs placed periodically over a $N_s$-SC OFDM system along with a  frequency-sampling factor of $\eta_f > 1$, as described in Sec.~\ref{sec_sys_model}. If an illegitimate receiver (Eve) attempts to estimate the delay (TOF) $\tau_E$ of the target using the ML criterion, then it also observes \emph{ghost target delay estimates} given by
\begin{equation}\label{eq_delay_ghost_values}
	\hat{\tau}_{E,m} = \tau_E + \frac{m}{\eta_f \Delta f},  \quad m \in \mathcal{M}_g(\tau_E),
\end{equation}
where the index set $\mathcal{M}_g(\tau_E)$ is defined as $	\mathcal{M}_g(\tau_E) \triangleq$
\begin{equation} \label{eq_index_set_range}
	\left\{
	m \in \mathbb{Z} \setminus \{0\}
		\Big|
	\left\lceil
	-\eta_f\tau_E \Delta f
	\right\rceil \leq m \leq \left\lfloor
\eta_f(1-\tau_E \Delta f)
	\right\rfloor
	\right\}.
\end{equation}
\end{theorem}

\begin{proof}
Similar to proof of Theorem~\ref{thm_ghost_target_location}. Skipped for brevity.
\end{proof}
Thus, Theorem~\ref{thm_range_ambiguity} shows how for any true value of $\tau_E$ and the location of an Eve, by using sparse pilots that are placed periodically over the OFDM SCs, Eve will always observe multiple ghost range values of the target.

In summary, this section has demonstrated that sparse arrays and periodically placed sparse pilots induce controllable impairments in the range-angle MIMO AF observed at Eve. Through the aliasing introduced by SD and FD sub-sampling, multiple indistinguishable ghost estimates arise in both the angular and range dimensions, and protects the legitimate sensing operation. A remaining question to investigate is whether and how these ambiguities can also translate in securing the target positioning itself, and we answer this next. 

\section{Characterization of Positioning Privacy via Sparse Arrays and Sparse Pilots}
Given estimates of the DoD, DoA, and range, Eve can exploit the sensing geometry to infer the target location and ultimately determine its position coordinates. The next result establishes a sufficient condition under which sparse arrays and sparse pilots can impose ambiguity on the positioning process.

\begin{theorem}\label{thm_localization_ambiguity}
Consider a CC-ISAC system described in Sec.~\ref{sec_sys_model} with target parameters $(\theta_0,\theta_E,\tau_E)$. Let $\mathbf{p}_{BS}, \mathbf{p}_E \in \mathbb{R}^2$ denote the locations of the BS and Eve, respectively. A sufficient condition for positioning ambiguity at Eve is that there exist $n,m \in \mathbb{Z} \setminus \{0\}$ and $\beta \in \mathbb{R}_+$ satisfying
\begin{equation}\label{eq_localization_amb_eqn}
	\mathbf{e}_2^T \big(\mathbf{p}_E + \beta \mathbf{u}_E - \mathbf{p}_{BS}\big)
	=
	\left(\sin(\theta_0) + \frac{2n}{\eta_s}\right)
	\left(c\tau_E + \frac{c m}{\eta_f \Delta f} - \beta\right),
\end{equation}
where $\mathbf{u}_E = [\cos\theta_E, \sin\theta_E]^T$ and $\mathbf{e}_2 = [0,1]^T$. Under this condition, Eve cannot uniquely determine the target location.
\end{theorem}

\begin{proof}
See Appendix~\ref{app_Position_AMB_Cond}.
\end{proof}

We illustrate Theorem~\ref{thm_localization_ambiguity} pictorially in Fig.~\ref{fig_localization_figure}. The theorem reveals the role of sparse arrays in positioning privacy. Range-domain ambiguity alone does not guarantee positioning ambiguity; otherwise,  the true target position can still be identified by triangulating the uniquely estimated DoA and DoD with the candidate range-ghost estimates. A similar argument holds when only spatial-domain ambiguity is present. Thus, positioning privacy emerges through the presence of both range- and SD ambiguities. Next, we characterize the conditions under which the criterion in Theorem~\ref{thm_localization_ambiguity} is satisfied.

\vspace{-0.2cm}
\begin{figure}
	\vspace{-0.6cm}
	\centering
	\includegraphics[width=0.95\linewidth]{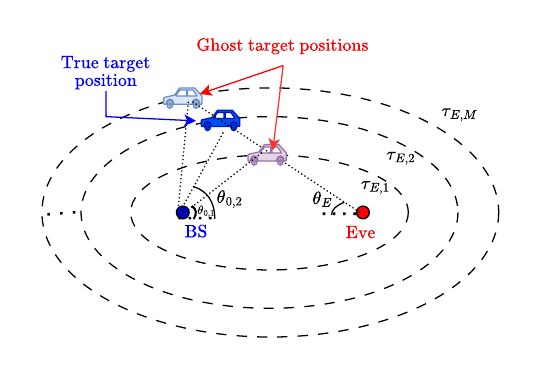}
	\caption{Illustration of Theorem~\ref{thm_localization_ambiguity}: Different delay ambiguities generate multiple bistatic ellipses, while spatial aliasing at the BS produces multiple DoD rays. Since the DoA at Eve is uniquely estimated, all feasible target locations must lie on the same DoA ray. Ghost target locations arise whenever a DoD ray intersects a delay ellipse at a point lying on this DoA ray. The existence of such intersections is characterized by the non-zero integers $n,m$ in Theorem~\ref{thm_localization_ambiguity}.} 
	\vspace{-0.2cm}
	\label{fig_localization_figure}
\end{figure}

\begin{proposition}\label{prop_existance_of_position_error}
For sufficiently high spatial and frequency sub-sampling factors $\eta_s$ and $\eta_f$, there exists at least one pair $(n,m) \in  \mathbb{Z}^2$, $n \neq 0, m \neq 0$ that satisfies the condition in Theorem~\ref{thm_localization_ambiguity}, ensuring sensing privacy in CC-ISAC systems.
\end{proposition}
\begin{proof}
 See Appendix~\ref{app_Position_Amb_satisfy}.
\end{proof}

\section{Performance of Legitimate ISAC}\label{sec_legitimate_ISAC}
In this section, we analyze the impact of sub-sampling on the legitimate sensing performance at Alice and communication performance at a UE. To this end, we first present the range-angle MIMO-AF at Alice in the following result.
\begin{lemma}\label{lem_AF_Alice_expression}
For the ISAC system described in Sec.~\ref{sec_sys_model}, the joint range-angle MIMO ambiguity function at Alice, associated with target parameters $(\boldsymbol{\theta}_0,\boldsymbol{\theta}_E,\boldsymbol{\tau}_E)$
and evaluated at the reference parameters $(\theta_0',\theta_A',\tau_A')$, can be expressed as
\vspace{-0.1cm}
\begin{multline}
	\psi_A(\boldsymbol{\theta}_0,\theta_0',\boldsymbol{\theta}_A,\theta_A',\boldsymbol{\tau}_A,\tau_A')  = \\ \sum_{k=1}^{K_s}\left|\frac{\sin\left(\frac{\pi N_r}{2}\Delta^k_{\theta_A',\theta_A}\right)}{N_r\sin\left(\frac{\pi}{2}\Delta^k_{\theta_A',\theta_A}\right)}\right| \times \left|\frac{\sin\left({\pi N_s\Delta f}\Delta^k_{\tau_E',\tau_A}\right)}{N_s\sin\left({\pi\Delta f}\Delta^k_{\tau_E',\tau_A}\right)}\right|,
\end{multline}
where all symbols are as defined in Lemma~\ref{lem_AF_Eve_expression}.
\end{lemma}
\begin{proof}
See Appendix~\ref{app_Alice_AF}.
\end{proof}
Importantly, we observe from Lemma~\ref{lem_AF_Alice_expression} that the overall range-angle MIMO AF at Alice is independent of the sub-sampling factors $\eta_s$ and $\eta_f$, thanks to the monostatic structure of Alice, and the pilot- and data-based sensing of targets. 

Next, at the communication UE, assuming that the channel correlation across SCs remains the same across every pair of transmit-receiver antenna, we have the following result.
\begin{lemma}\label{lemma_CSI_error_variance}
	Consider the system model described in Sec.~\ref{sec_sys_model}. Let
	$\boldsymbol{\Sigma}$ denote the spatial covariance matrix of $\mathbf{H}_{c,n}$, i.e.,
	\begin{equation}
	\boldsymbol{\Sigma}
	=
	\mathbb{E}\!\left[
	\mathrm{vec}(\mathbf{H}_{c,n})
	\mathrm{vec}(\mathbf{H}_{c,n})^H
	\right] \ \in\mathbb{C}^{N_cN_t\times N_cN_t},
\end{equation}
	and let $r[m,\!n]$,\! $m,n\!\in\!\mathcal{N}_s$, be the correlation coefficient between the channels on SCs $m$, $n$. Assume that the channels on the pilot SCs are estimated using the LS method along with $L=N_t$ OFDM symbols. Assume furthermore that these LS estimates are used to obtain the LMMSE channel estimate, $\hat{\mathbf{h}}_{c,n}^{v}\triangleq \mathrm{vec}(\hat{\mathbf{H}}_{c,n})$ on a data SC-$n\in\mathcal{N}_s^d$. Then, the error covariance matrix of $\hat{\mathbf{h}}_{c,n}^{v}$ is formulated as
		\vspace{-0.1cm}
	\begin{equation}
			\vspace{-0.1cm}
		\mathbf{E}_n
		=
		r(0)\boldsymbol{\Sigma}
		-
		(\mathbf{r}_{np}^{T}\otimes\boldsymbol{\Sigma})
		\left(
		\mathbf{R}_{pp}\otimes\!\boldsymbol{\Sigma}
		+
		\sigma_w^2
		\mathbf{I}_{N_cN_tN_s^p}
		\right)^{-1}\!\!
		(\mathbf{r}_{np}^{*}\otimes\boldsymbol{\Sigma}),
	\end{equation}
	where we have $r(0)\triangleq r[0,0]$, and
		\vspace{-0.1cm}
	\begin{equation}
	\mathbf{r}_{np} \in \mathbb{C}^{N_s^p}: [\mathbf{r}_{np}]_i = r[n,m_i], \qquad i\in[N_s^p],
\end{equation}
	with $m_i$ denoting the $i$th pilot SC, represents the correlation between the $n$th data SC and all pilot SCs. Furthermore,
		\vspace{-0.1cm}
	\begin{equation}
			\vspace{-0.1cm}
	\mathbf{R}_{pp} \in \mathbb{C}^{N_s^p \times N_s^p}: 	[\mathbf{R}_{pp}]_{i,j}=r[m_i,m_j]
\end{equation}
	is the correlation matrix of all the pilot SCs.
\end{lemma}

\begin{proof}
	Skipped due to space constraints.
\end{proof}
We are now ready to state the main result of this section.
\begin{theorem}\label{prop_legitimate_performance}
For the CC-ISAC system described in Sec.~\ref{sec_sys_model} having SD and FD sub-sampling factors $\eta_s$ and $\eta_f$ and a $N_t$-antenna BS and $N_s$-SC MIMO-OFDM-ISAC system,
\begin{enumerate}[leftmargin=*]
	\item The sensing performance at the legitimate receiver (Alice) remains identical to that of the Nyquist-sampled system, i.e., with $\eta_s=\eta_f=1$, and 
	\item The achievable communication throughput scales as
		\vspace{-0.1cm}
	\begin{equation}\label{eq_Rate_Comm_UE}
		R =
		\mathcal{O}\!\left\{
		\left(1-\frac{1}{\eta_f}\right)
		\sum_{k=1}^{\mathrm{rank}(\hat{\mathbf{H}}_{c,n})}\!\!\!\!
		\log_2\!\left(
		1+\frac{P}{N_tN_s}\mu_k
		\right)
		\right\},
	\end{equation}
	where $\{\mu_k\}_k$ denote the eigenvalues of the average Gram matrix $\frac{1}{N_s^d}
		\sum_{n\in\mathcal{N}_s^d}
		\mathbf{G}_{c,n}\mathbf{G}_{c,n}^{H}$. Here, we have $	\mathbf{G}_{c,n}=$
			\vspace{-0.2cm}
	\begin{equation}
		\left(
		\frac{1}{N_t}
		\sum_{i=1}^{N_t}
		(\mathbf{e}_i^{T}\!\otimes\!\mathbf{I}_{N_r})
		\mathbf{E}_n
		(\mathbf{e}_i^{T}\!\otimes\!\mathbf{I}_{N_r})^{H}
		+
		\sigma_w^2\mathbf{I}_{N_r}
		\right)^{-\frac{1}{2}}
		\hat{\mathbf{H}}_{c,n},
	\end{equation}
	where $\mathbf{e}_i\in\mathbb{C}^{N_t}$ is the $i$th canonical basis vector, $\hat{\mathbf{H}}_{c,n}$ is the LMMSE estimate of $\mathbf{H}_{c,n}$, and $\mathbf{E}_n$ is the estimation error covariance matrix characterized in Lemma~\ref{lemma_CSI_error_variance}.
\end{enumerate}

\end{theorem}

\begin{proof}
	See Appendix~\ref{app_legitimate_performance}.
\end{proof}
We now make the following important observations: 
\begin{enumerate}[leftmargin=*]
	\item \textbf{Frequency-domain sub-sampling:}
	Theorem~\ref{prop_legitimate_performance} reveals the following fundamental trade-off: On one hand, increasing $\eta_f$ improves sensing privacy and reduces pilot overhead, thereby increasing the pre-log factor in~\eqref{eq_Rate_Comm_UE}, which enhances the achievable communication rate. On the other hand, a higher $\eta_f$ degrades channel prediction accuracy on data SCs (particularly for those farther from the pilot SCs), since the inter-channel correlation decreases with distance between SCs. Specifically, we can argue that the $\|\mathbf{E}_n\|_F$ in Lemma~\ref{lemma_CSI_error_variance} increases with $\eta_f$, which in turn reduces the effective gains $\{\mu_k\}_k$ in~\eqref{eq_Rate_Comm_UE}. Therefore, the optimal choice of $\eta_f$ depends on the frequency selectivity of the UE's channel. Importantly, this trade-off, however, exists even in the absence of Eve and it is inherent to OFDM communication. Thus, the proposed sensing-privacy mechanism does not introduce any additional trade-offs beyond those already present in fully-private legitimate ISAC systems.
	\item \noindent\textbf{Spatial-domain sub-sampling:}
	Increasing $\eta_s$ extends the inter-antenna spacing and consequently reduces spatial correlation. As noted in~\cite{Shiu_TCOM_2000}, lower spatial correlation is generally beneficial for achieving a higher communication throughput. However, deriving a precise throughput scaling law with respect to $\eta_s$ is beyond the scope of this paper.
\end{enumerate}

\begin{remark}
The proposed framework also extends to the complementary setting, where Alice operates in a multistatic configuration, while Eve is located in close proximity to the BS and therefore observes an approximately monostatic geometry. In this case, Alice can cooperate with the ISAC BS to jointly process pilot and data SCs for target sensing, thereby resolving range-domain ambiguities and uniquely positioning the target despite DoD ambiguities. By contrast, although Eve's near-monostatic geometry may largely eliminate DoD ambiguities, it still suffers from range-domain ambiguities, leading to substantial positioning errors. Hence, the proposed sensing-privacy framework remains effective, regardless of the radar geometries available to Alice and Eve.
\end{remark}
In the next section, we perform Monte-Carlo simulations and numerically characterize our findings. 

\vspace{-0.2cm}
\section{Numerical Results}
Unless stated otherwise, the system parameters are set as follows: the BS, Alice, and Eve employ $N_t=N_r=N_e=16$ antennas, while the communication UE has $N_c=4$ antennas. The OFDM setup consists of $N_s=512$ SCs with the SC BW given by $\Delta f = 120$ kHz. The target DoDs are sampled as $\theta_0 \sim \mathcal{U}[-60^\circ,60^\circ]$, and the delays are sampled from $\mathcal{U}(0,1/\Delta f)$. The transmit power and noise variance are $P=15$ dBm and $\sigma^2=-110$ dBm, respectively. The link's pathloss is modeled as $\beta_{PL} = C_0 d^{-2}$, where $C_0=-30$ dB is the reference pathloss and $d$ is the link distance. Furthermore, the channels at the UE undergo Rayleigh fading with $\boldsymbol{\Sigma} = \mathbf{I}_{N_cN_t}$. The sensing performance is evaluated using the root-mean-square error (RMSE) of the estimated target parameter, while the communication performance is evaluated using the sum-rate.
\vspace{-0.15cm}
\subsection{Spatial Transmit-MIMO Ambiguity Function at Eve}
\begin{figure}[t]
	\vspace{-0.8cm}
\centering
\hspace{-0.5cm}
\begin{subfigure}{0.49\linewidth}
	\centering
	\includegraphics[width=1.16\linewidth]{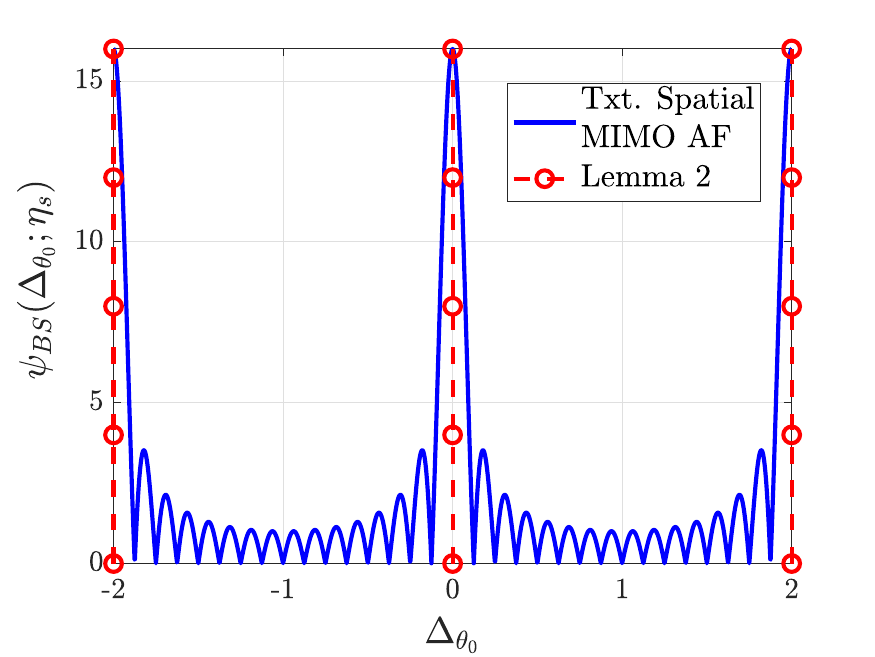}
	\caption{$\eta_s=1$.}
	\label{fig_AF_eta_1}
\end{subfigure}
\hspace{0.05cm}
\begin{subfigure}{0.49\linewidth}
	\centering
	\includegraphics[width=1.16\linewidth]{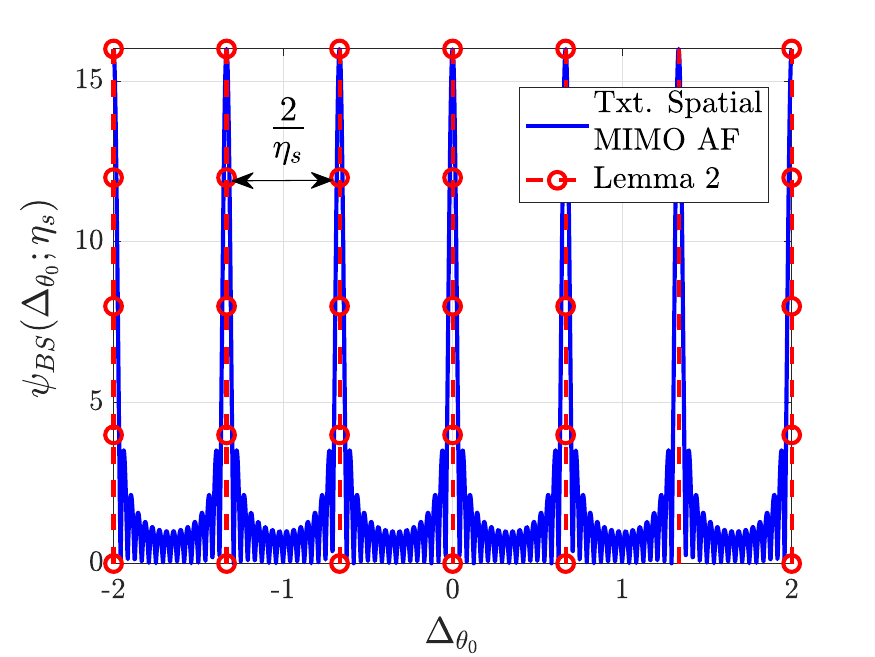}
	\caption{$\eta_s=3$.}
	\label{fig_AF_eta_3}
\end{subfigure}
\caption{Transmit Spatial MIMO-AF with $N_t=16$.}
\label{fig_spatial_AF}
\vspace{-0.3cm}
\end{figure}
In Fig.~\ref{fig_AF_eta_1} and Fig.~\ref{fig_AF_eta_3}, we plot the spatial transmit MIMO-AF observed at Eve for spatial sampling factors of $\eta_s=1$ and $\eta_s=3$, respectively. When $\eta_s=1$, which corresponds to a Nyquist-sampled array at the BS, the AF exhibits a single dominant peak at $\Delta_{\theta_0}=0$ along with rapidly decaying sidelobes. Consequently, within the valid range $\Omega_s$, the spatial signature of a target remains unambiguous, enabling even an arbitrary Eve to sense the presence of targets. By contrast, when $\eta_s\!>\!1$, e.g., $\eta_s=3$, corresponding to an undersampled USA, multiple peaks appear periodically in the AF with identical magnitudes. These additional peaks introduce ambiguity in the transmit spatial signature of a target, even in noiseless cases, and form the basis for spatial-sensing privacy against Eve. Furthermore, Fig.~\ref{fig_AF_eta_3} shows that the peak locations coincide with those predicted by Lemma~\ref{lem_ambiguity_function}, which validates our analysis.
\vspace{-0.2cm}
\subsection{Range-Ambiguity Function at Eve}
\begin{figure}[t]
	\centering
	\vspace{-0.1cm}
	\hspace{-0.35cm}
	\begin{subfigure}{0.49\linewidth}
		\centering
		\includegraphics[width=1.15\linewidth]{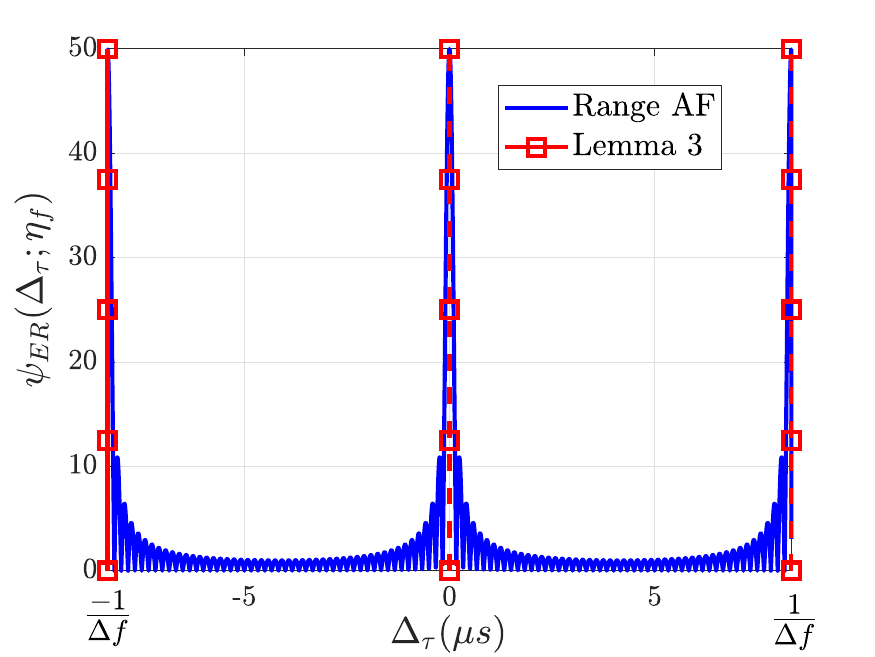}
		\caption{$\eta_f=1$.}
		\label{fig_range_AF_eta_1}
	\end{subfigure}
	\hspace{0.05cm}
	\begin{subfigure}{0.49\linewidth}
		\centering
		\includegraphics[width=1.15\linewidth]{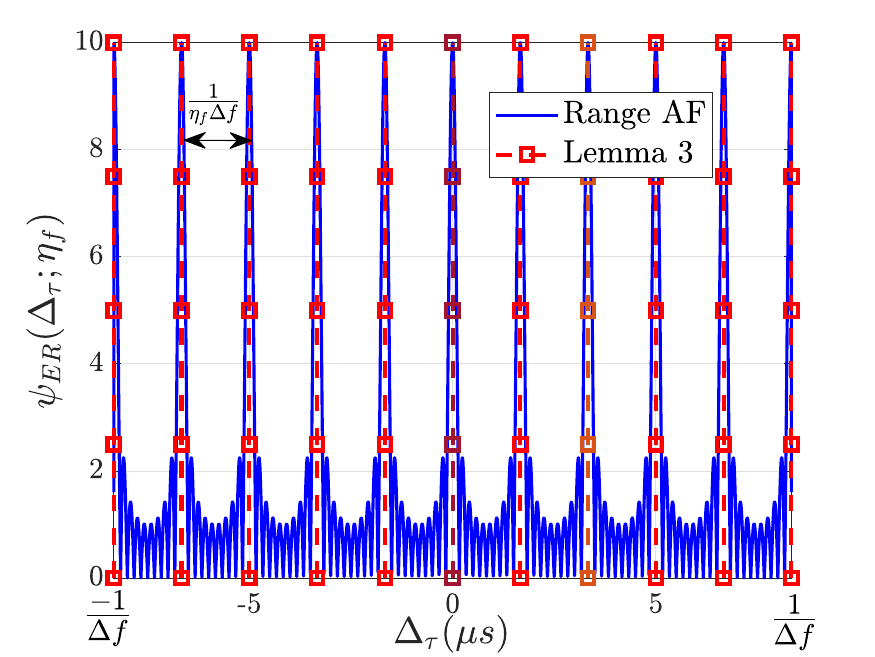}
		\caption{$\eta_f=5$.}
		\label{fig_range_AF_eta_5}
	\end{subfigure}
	\caption{Range AF at Eve with $B=6$ MHz and $\Delta f=120$ kHz.}
	\label{fig_range_AF}
	\vspace{-0.4cm}
\end{figure}
In Fig.~\ref{fig_range_AF_eta_1} and Fig.~\ref{fig_range_AF_eta_5}, we plot the range-AF observed at Eve for $N_s=50$, and frequency sampling factors of $\eta_f=1$ and $\eta_f=5$, respectively. When $\eta_f=1$, the range AF exhibits a single dominant peak at $\Delta_\tau=0$, enabling Eve to obtain an unambiguous estimate of the target range. However, when $\eta_f$ increases to $5$, the AF develops multiple peaks whose magnitudes are equal to that of the peak at $\Delta_\tau=0$. As a result, Eve can no longer reliably distinguish the true target range from its ghost counterparts, thereby increasing the likelihood of incorrect range estimation, as it will be demonstrated later. Furthermore, it can be observed that the peak value of the AF decreases as $\eta_f$ increases. Specifically, the peak magnitude reduces by a factor of $\eta_f$. This behaviour arises because increasing $\eta_f$ reduces the number of pilot SCs available for sensing. Consequently, Eve has fewer observations over which coherent matched-filter combining can be performed, leading to a reduction in the achievable processing gain.

Therefore, frequency sub-sampling has two important effects on the range AF at Eve: $(i)$ it introduces multiple ambiguity peaks that are as strong as the peak corresponding to the true target TOF, and $(ii)$ it reduces the coherent processing gain available for range estimation. Together, these effects significantly degrade Eve's range-sensing capability. Even if Eve were to select the correct peak among the multiple candidates, the reduced processing gain would still lead to a less accurate range estimate. Finally, it is evident from both figures that the theoretical predictions in Lemma~\ref{lem_range_AF_peaks} match the simulation results, validating the accuracy of our analysis.
\vspace{-0.2cm}
\subsection{Spatial-domain Ghost Targets and Performance}
We now characterize the achievable performance of the spatial-sensing at Alice, Eve, and the communication UE as a function of the spatial sampling factor at the BS array. 
\subsubsection{Spatial-domain Ghost-target Maps}
\begin{figure}
	\vspace{-0.85cm}
	\centering
	\includegraphics[width=\linewidth]{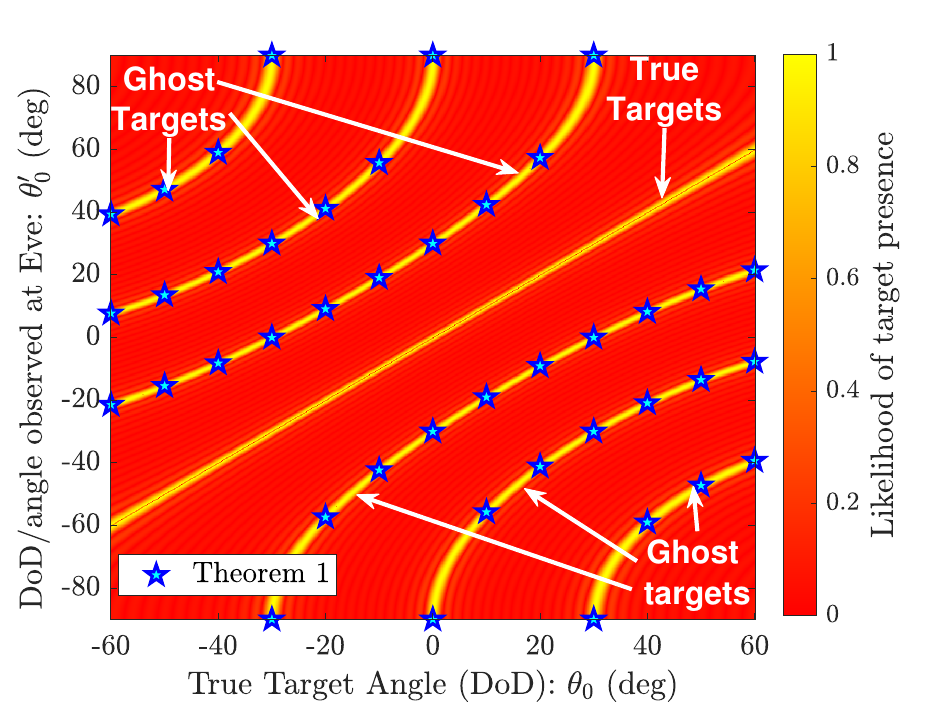}
	\caption{Observed angle at Eve vs. true angle for $\eta_s=4$.}
	\label{fig_trajectory}
	\vspace{-0.5cm}
\end{figure}
In Fig.~\ref{fig_trajectory}, we present a heatmap of the DoDs observed by Eve for all possible true target DoDs when $\eta_s=4$. The diagonal line corresponds to true/correct DoD angles, while several additional high-likelihood response curves appear to be away from the diagonal, representing ghost targets as discussed in Fig.~\ref{fig_AF_eta_3}. For every true target angle $\theta_0$, multiple candidate directions are observed at Eve, i.e., each vertical slice contains several normalized peaks of the matched filter response, and importantly, ghost targets are observed for every possible true target angle. The star markers indicate the theoretical ghost-target locations predicted by Theorem~\ref{thm_ghost_target_location}. The close agreement between these markers and the peaks validates the analytical results.

\subsubsection{Performance of Spatial-domain ISAC}
We now evaluate the system-level performance of the proposed sparse-array-enabled CC-ISAC system, averaged over multiple realizations of the target, UE, and Eve locations. ML-based sensing is employed at both Alice and Eve, and the corresponding results are shown in Fig.~\ref{fig_Angle_domain_ISAC_perf}. On this figure, the left axis shows the sensing RMSE of estimating the target DoD, and the right axis shows the communication sum-rate at the UE. These are evaluated for different spatial sampling factors $\eta_s$, and for two transmit power levels of $P=5$ and $P=15$ dBm, which determine the operating SNR. The results reveal several key observations: $(i)$ When $\eta_s$ increases, the RMSE at Eve increases monotonically due to the growing number of ghost targets introduced by spatial undersampling, which creates ambiguity in estimating the true target direction. By contrast, the sensing RMSE at Alice remains unchanged for all values of $\eta_s$. This is because, under the monostatic setting, Alice estimates the target direction purely from receive array processing, whose AF is unaffected by the sparsity of the BS transmit array; $(ii)$ The RMSE at Alice decreases upon increasing the power, indicating that its estimation error is primarily noise-limited. By contrast, Eve's RMSE remains nearly identical across different SNR levels, showing that its estimation error is dominated by the ambiguity induced by the sparse transmit array rather than by noise. So, increasing the SNR does not improve Eve's sensing performance when the BS employs sparse arrays; $(iii)$ The communication sum-rate is constant for all values of $\eta_s$.
This is because, the spatial undersampling at BS does not alter the eigenvalue distribution of i.i.d. MIMO channels and hence the communication throughput. 
\begin{figure}[t]
	\vspace{-0.85cm}
	\centering
	\includegraphics[width=1\linewidth]{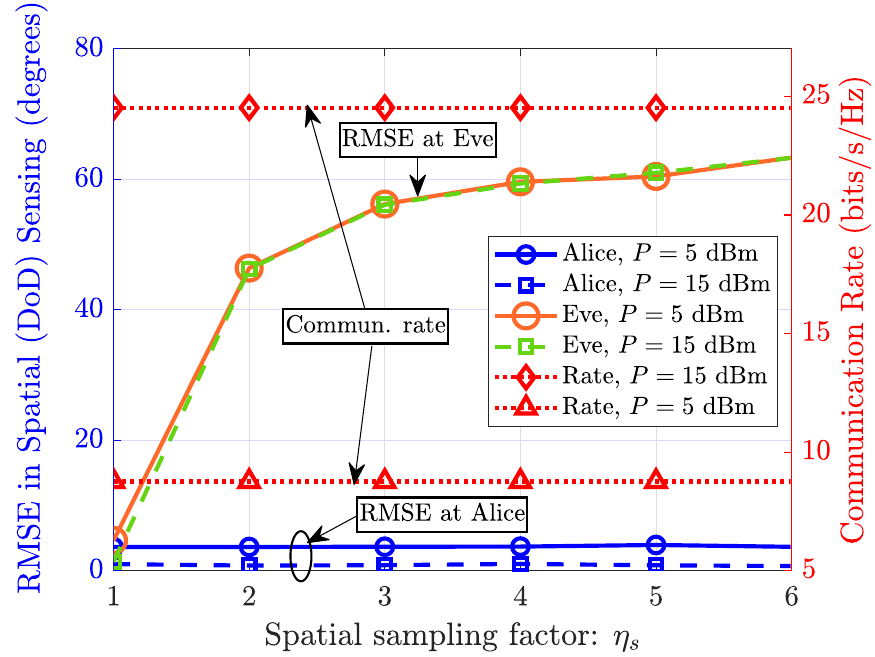}
\caption{Spatial-domain ISAC performance vs. $\eta_s$.}
\label{fig_Angle_domain_ISAC_perf}
\vspace{-0.2cm}
\end{figure}
\vspace{-0.2cm}
\subsection{Range-domain Ghost Targets and Performance}
\subsubsection{Range-domain Ghost-target Maps}
\begin{figure}
	\vspace{-0.35cm}
	\centering
	\includegraphics[width=\linewidth]{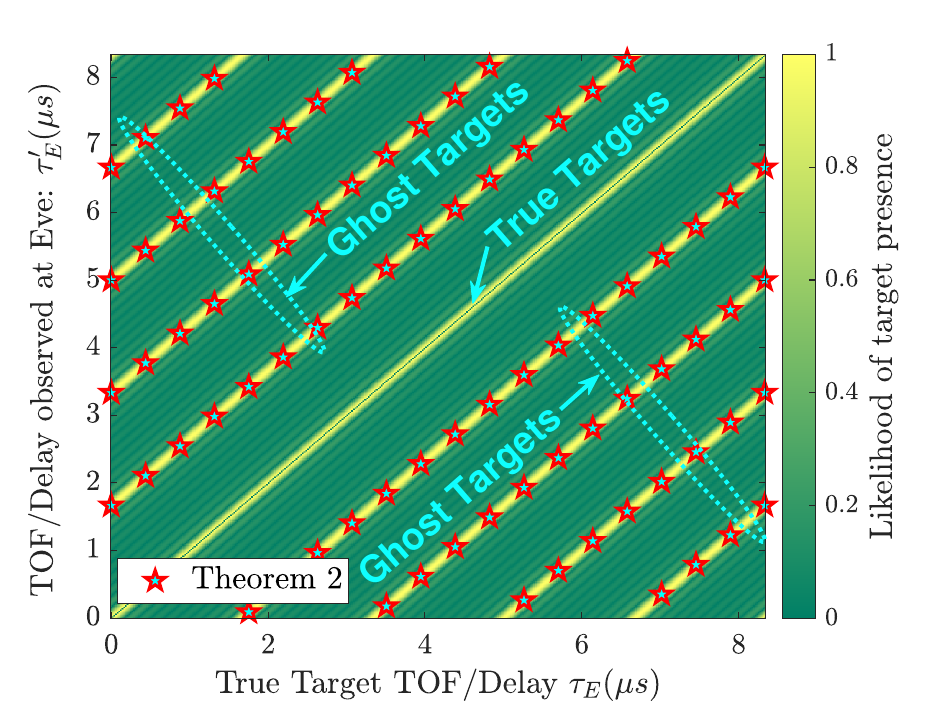}
	\caption{Observed delays at Eve vs. true delay for $\eta_f=5$.}
	\label{fig_trajectory_delay}
		\vspace{-0.5cm}
\end{figure}
In Fig.~\ref{fig_trajectory_delay}, we plot the TOF observed at Eve as a function of the true TOF when $\eta_f=5$. Similar to Fig.~\ref{fig_trajectory}, it can be observed that, for each true TOF value, Eve encounters multiple peak likelihood values that correspond to different candidate TOF estimates along the vertical slice. Thus, when pilot SCs are periodically interleaved with data SCs as shown in Fig.~\ref{fig_OFDM_structure}, Eve inevitably observes multiple ghost range estimates in addition to the true target range. The star markers in the figure denote the locations of the ghost-TOF estimates predicted by Theorem~\ref{thm_range_ambiguity}. It is also evident that these analytically predicted locations align with the simulated ghost targets, validating the accuracy of the analysis. Moreover, the number of ghost TOF estimates scales linearly with the frequency sampling factor $\eta_f$, confirming that higher pilot sparsity leads to severe range ambiguity at Eve.

\subsubsection{Performance of Range-domain ISAC}
\begin{figure}
	\vspace{-0.85cm}
	\centering
	\includegraphics[width=1\linewidth]{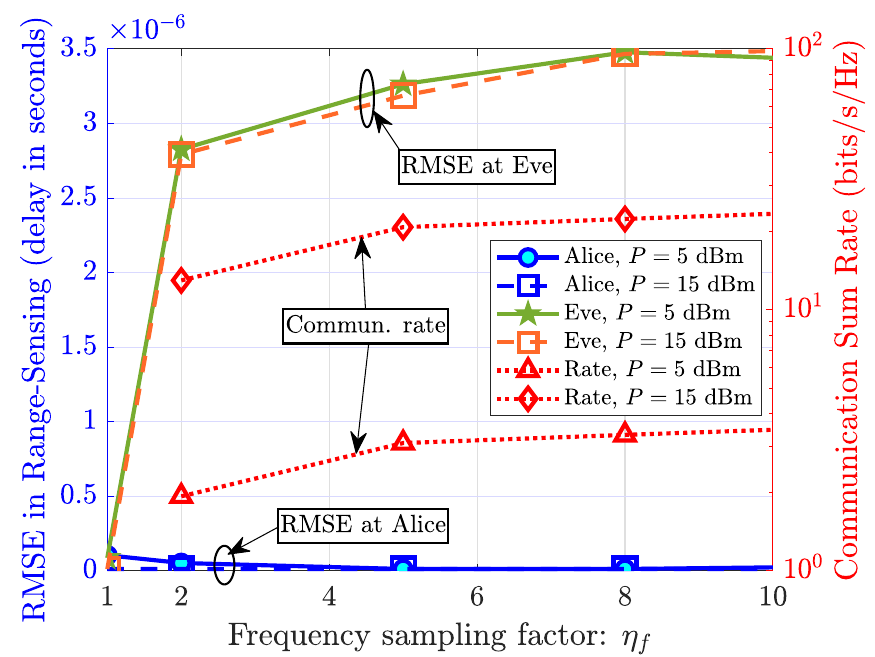}
	\caption{Range-domain ISAC performance vs. $\eta_f$.}
	\label{fig_Range_domain_ISAC_perf}
		\vspace{-0.5cm}
\end{figure}
Next, in Fig.~\ref{fig_Range_domain_ISAC_perf}, we evaluate the performance of the proposed range-domain ISAC framework. Similar to Fig.~\ref{fig_Angle_domain_ISAC_perf}, the left y-axis depicts the RMSE of the TOF estimate at Alice and Eve, while the right y-axis shows the communication sum-rate, both as functions of the frequency sampling factor $\eta_f$. At Eve, the RMSE in TOF-estimation increases with $\eta_f$ due to the introduction of ghost targets in the range domain. 
Moreover, the RMSE remains largely insensitive to the operating SNR, indicating that ambiguity-induced errors dominate the sensing errors encountered by Eve. This behaviour is similar to the SD sensing results from Fig.~\ref{fig_Angle_domain_ISAC_perf}. By contrast, Alice consistently achieves lower TOF-estimation errors. Although increasing $\eta_f$ alters the number of pilot and data SCs, the total number of SCs remains the same and Alice exploits all of them for sensing. Consequently, the range-sensing performance at Alice remains invariant to $\eta_f$ and continues to operate in the conventional noise-limited regime.

From the communication perspective, the achievable sum-rate increases with $\eta_f$, since a higher fraction of SCs is available for data transmission. However, the throughput gains gradually saturate as $\eta_f$ increases, because most of the benefits offered by the pre-log factor are already realized at moderate pilot sparsity levels. This is consistent with Theorem~\ref{prop_legitimate_performance}. Furthermore, for the range of $\eta_f$ considered in this figure, the pilot density remains sufficient for ensuring accurate channel prediction on data SCs, and hence the throughput degradation caused by channel estimation errors remains negligible. As a result, the increase in the pre-log factor dominates, rendering the achievable rate effectively non-decreasing with $\eta_f$. Overall, increasing $\eta_f$ enhances sensing privacy by aggravating the range ambiguity perceived by Eve, while preserving the communication throughput achievable by the legitimate system.
\vspace{-0.2cm}
\subsection{Performance of Positioning and Communications}
\subsubsection{Ghost-Positions Map}
\begin{figure}
	\vspace{-0.85cm}
\centering
\includegraphics[width=\linewidth]{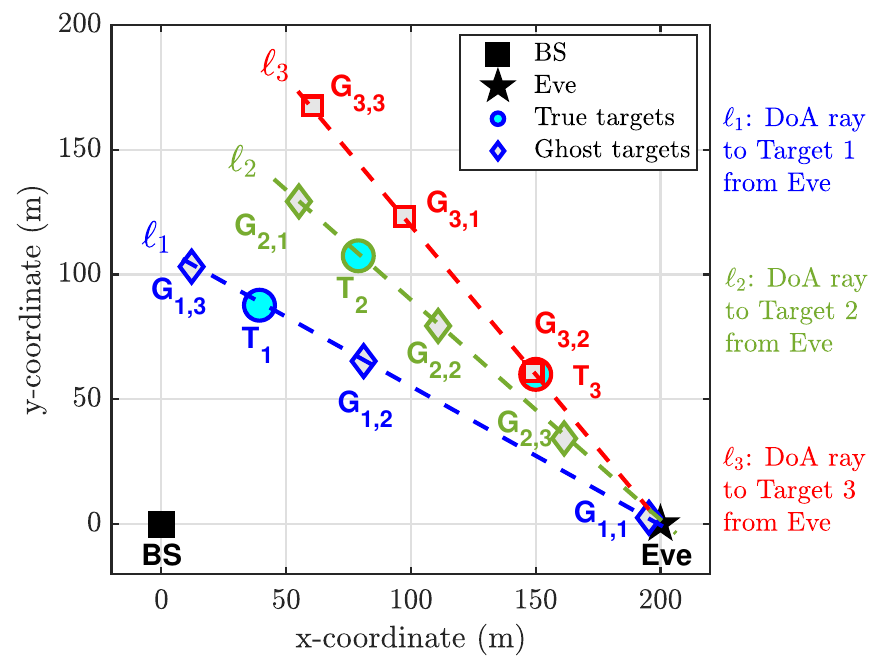}
\caption{Positions of BS, Eve, and ghost targets for $K_s=3$.}
\label{fig_ghost_target_locations}
	\vspace{-0.5cm}
\end{figure}
We now illustrate how the joint spatial-and range-domain ambiguities translate into positioning ambiguity at Eve for $\eta_s=6$ and $\eta_f=64$, when the BS and Eve are located at $\mathbf{p}_{BS} = [0,0]^T$, and $\mathbf{p}_{E} = [200,0]^T$, respectively. Specifically, we plot the true locations and their corresponding ghost locations for $K_s=3$ targets. In Fig.~\ref{fig_ghost_target_locations}, $T_i$, $i\in\{1,2,3\}$, denotes the true location of the $i$th target, while $G_{i,j}$, $j\in\{1,2,3\}$, denotes the $j$th ghost location associated with the $i$th target. It is evident that multiple ghost positions emerge for each true target location, and all of them correspond to Eve's observed spatial-range domain sensing parameters (note that all ghost locations of a target lie on the same DoA ray corresponding to that target). Consequently, Eve cannot uniquely identify the true target position, thereby validating the positioning ambiguity characterized in Theorem~\ref{thm_localization_ambiguity}. Thus the proposed signaling framework can indeed provide positioning-level privacy against unauthorized sensing.
\subsubsection{Performance of Positioning-based ISAC}
In this final study, we evaluate the overall system-level performance as a function of the SD and FD sub-sampling factors in Fig.~\ref{fig_positioning_figure}. The communication channel across SCs is modeled using an exponential correlation model with base $0.9$~\cite{Vineeth_TWC_2021}. The results are averaged over multiple target realizations whose true range and DoD measured from the BS are uniformly distributed in $[50,150]$ meters and $[10^\circ,80^\circ]$, respectively.

Using the same plotting convention as in Figs.~\ref{fig_Angle_domain_ISAC_perf} and~\ref{fig_Range_domain_ISAC_perf}, we observe that, for a fixed value of $\eta_s$, the positioning RMSE at Eve increases with $\eta_f$. This behavior arises because increasing $\eta_f$ creates additional delay ambiguities, causing Eve to associate the target with multiple bistatic ellipses. Since the number of ghost DoD rays remains fixed for a given $\eta_s$, the larger set of ellipses generates a higher number of candidate intersection points, thereby increasing the positioning error. Now, when $\eta_s$ increases, the number of ghost DoD rays further increases, which in turn expands the set of feasible ghost target locations obtained from the intersections between the DoD rays and bistatic ellipses. Consequently, the positioning RMSE increases with both spatial and frequency sub-sampling. However, compared to Fig.~\ref{fig_Range_domain_ISAC_perf}, the onset of significant positioning errors occurs at slightly larger values of $\eta_f$. This is because positioning ambiguity requires sufficient ambiguities simultaneously in both the angle and range domains, as established in Theorem~\ref{thm_localization_ambiguity}. In accordance with Proposition~\ref{prop_existance_of_position_error}, such ambiguities emerge whenever at least one of the sub-sampling factors is sufficiently high, with the severity of the positioning error increasing when both $\eta_s$ and $\eta_f$ are large.

By contrast, Alice remains largely unaffected by the choice of sub-sampling factors, in line with Theorem~\ref{prop_legitimate_performance}. Since Alice can unambiguously estimate both the target angle and delay, it does not experience any noticeable degradation in positioning accuracy. From a communication perspective, the achievable throughput initially increases with $\eta_f$ due to the reduction in pilot overhead, but eventually decreases as sparse pilot allocations degrade the accuracy of channel prediction on data SCs. Nevertheless, even when operating close to the maximum achievable throughput, the proposed framework creates a positioning privacy gap of approximately $40$ meters. Furthermore, by operating at a slightly lower rate (roughly $10$ bits/s/Hz below the maximum), the privacy gap can be further increased to nearly $80$ meters. Importantly, this apparent communication-privacy trade-off is inherent to the native OFDM protocol and it is present even without an Eve (see discussions after Theorem~\ref{prop_legitimate_performance}). Thus, the proposed sensing-privacy mechanism does not introduce additional trade-offs beyond those that are already inherent in native legitimate ISAC operation. 
\vspace{-0.1cm}

\begin{figure}
	\vspace{-0.87cm}
	\centering
	\includegraphics[width=\linewidth]{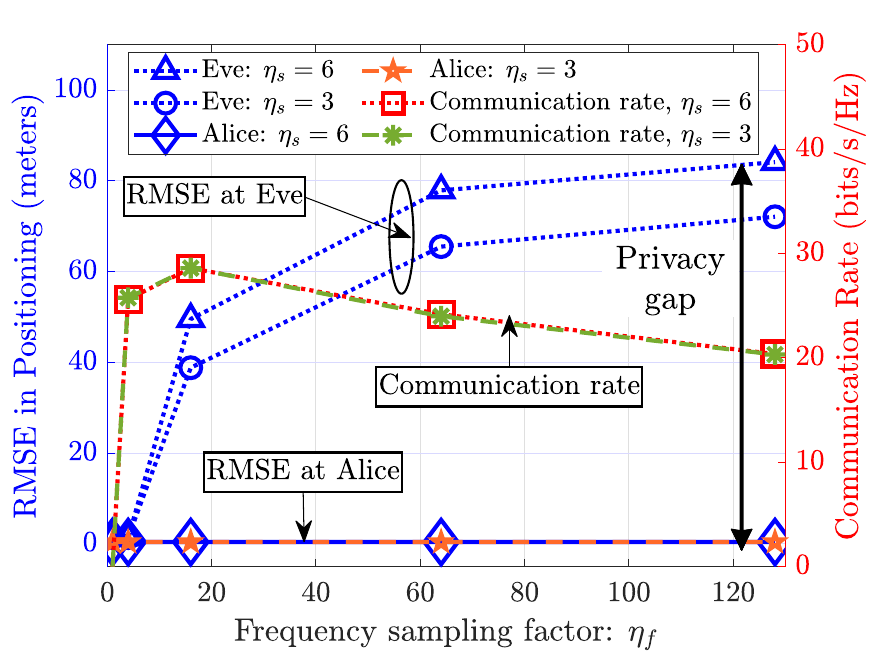}
	\caption{Positioning performance in ISAC vs. $\eta_f$, $\eta_s$.}
	\label{fig_positioning_figure}
		\vspace{-0.4cm}
\end{figure}
\section{Conclusions}
A sub-sampling based sensing-privacy framework was proposed for ISAC systems that exploits sparse arrays and sparse pilot allocations to induce controlled aliasing in both the SD and FD. We derived the range-angle MIMO AF at an unauthorized receiver and showed that SD and FD undersampling naturally generate ghost targets in the angle and range domains. We then established conditions under which these ambiguities translate into positioning ambiguity, thereby preventing an illegitimate receiver from unambiguously positioning a target. We further showed that the proposed framework preserves the sensing performance of the legitimate receiver and the native communication throughput of the ISAC system. Numerical results validated the analysis and demonstrated that spatial-frequency sub-sampling can effectively promote sensing privacy through deception by aliasing. Overall, the proposed approach reveals a new role for sparse arrays and sparse pilots in ISAC, transforming aliasing from a traditional limitation into a useful mechanism for physical-layer sensing privacy. Future work may include extending the ideas to mobility scenarios, non-uniform sparse arrays, mmWave bands, etc.
\vspace{-0.1cm}
\begin{appendices}
	\renewcommand{\thesectiondis}[2]{\Alph{section}:}
	\section{Proof of Lemma~\ref{lem_AF_Eve_expression}}\label{app_Eve_AF}
	Using~\eqref{eq_Eve}, we note that the impulse response of a MF referenced with respect to $(\theta_0',\theta_E',\tau_E')$ can be written as
	\vspace{-0.1cm}
	\begin{equation}\label{eq_MF_IR}
		\mathbf{G}_{e,MF} = \tilde{\mathbf{a}}_{N_e}(\theta_E')\Big\{\left(\mathbf{a}^H_{N_t}(\theta_0';\eta_s)\mathbf{X}_e\right) \odot \left(\mathbf{d}_e(\tau_E')\otimes \mathbf{1}_L\right)^T\Big\}.
		\vspace{-0.1cm}
	\end{equation}
	For notational compactness, let $\tilde{\mathbf{d}}_e(\tau)\triangleq \mathbf{d}_e(\tau)\otimes\mathbf{1}_L$.
	Then the output of the matched filter at Eve can be obtained as
	\vspace{-0.1cm}
	\begin{equation}\label{eq_MF_output}
		z_{e,MF} \triangleq \frac{1}{N^p_sL}\sum\nolimits_{q=1}^{N^p_sL}\mathbf{g}^H_{e,MF,q}\mathbf{y}_{e,q},
		\vspace{-0.1cm}
	\end{equation}
	where $\mathbf{g}_{e,MF,\ell}$ is the $\ell$th column of $\mathbf{G}_{e,MF}$ and $\mathbf{y}_{e,\ell}$ is the $\ell$th column of $\mathbf{Y}_e$. We can simplify~\eqref{eq_MF_output} as given in~\eqref{eq_MF_output_pre_long}-\eqref{eq_MF_output_long}, where $\mathbf{n}_{e,q}$ is the $q$th column of $\mathbf{N}_e$ and $\tilde{z}_{N,e}$ is the noise in the output of the MF.
	\begin{figure*}
		\vspace{-0.75cm}
		\begin{align}
			\!\!z_{e,MF} &=\! \frac{1}{N^p_sL}\!\sum_{q=1}^{N_s^pL}\!\sum_{k=1}^{K_s} \!\bar{\beta}^k_e \tilde{\mathbf{a}}^H_{N_e}(\theta_E')\tilde{\mathbf{a}}_{N_e}(\theta^k_E)\! \left[\mathbf{X}_e\right]^H_{:,q}\!\mathbf{a}_{N_t}(\theta_0';\eta_s)\mathbf{a}^H_{N_t}(\theta^k_0;\eta)\!\left[\mathbf{X}_e\right]_{:,q}\!\big[\tilde{\mathbf{d}_e}(\tau_E')\big]^*_q\big[\tilde{\mathbf{d}_e}(\tau^k_E)\big]_q\!\! + \!\frac{1}{N^p_sL}\!\sum_{q=1}^{N^p_sL}\! \mathbf{g}^H_{e,MF,q}\mathbf{n}_{e,q}\!\label{eq_MF_output_pre_long}\\
			&\hspace{-1cm} =\sum_{k=1}^{K_s}\bar{\beta}^k_e \underbrace{\tilde{\mathbf{a}}^H_{N_e}(\theta_E')\tilde{\mathbf{a}}_{N_e}(\theta^k_E)}_{\triangleq \tilde{z}^k_{1,e}} \times \underbrace{\frac{1}{N_s^pL}\sum_{n \in \mathcal{N}_s^p}\sum_{\ell=1}^{L} \mathbf{x}^H_{n,\ell}\mathbf{a}_{N_t}(\theta_0';\eta_s)\mathbf{a}^H_{N_t}(\theta^k_0;\eta_s)\mathbf{x}_{n,\ell}\big[{\mathbf{d}_e}(\tau_E')\big]^*_n\big[{\mathbf{d}_e}(\tau^k_E)\big]_n}_{\triangleq \tilde{z}^k_{2}} + \underbrace{\frac{1}{N_s^pL}\sum_{q=1}^{N_s^pL} \mathbf{g}^H_{e,MF,q}\mathbf{n}_{e,q}}_{\triangleq \tilde{z}_{N,e}} \label{eq_MF_output_long}
		\end{align}
		\vspace{-0.2cm}
		\hrule
		\vspace{-0.4cm}
	\end{figure*}
	We can clearly see from~\eqref{eq_MF_output_long} that the term $\tilde{z}^k_{1,e}$ characterizes the ambiguity in estimating $\theta^k_E$, whereas $\tilde{z}^k_{2}$, which captures the interaction between the ISAC waveform, the ISAC transmit array properties and the target delay value, determines the ambiguity in jointly estimating $\theta^k_0$ and $\tau^k_E$. Thus, the overall \emph{range-angle MIMO-AF} at Eve can be written as in~\eqref{eq_AF_intermediate_1}, which is further simplified in~\eqref{eq_AF_intermediate} at the top of the next page, where we first assumed that $L = \kappa N_t$ for some $\kappa \in \mathbb{Z}_+$; then upon exploiting the properties of the pilot sequences in~\eqref{eq_pilot_signal}, we used the following simplification:
	\vspace{-0.1cm}
	\begin{equation}\label{eq_pilot_SC_signal_OP}
		\sum\nolimits_{\ell=1}^{L}\mathbf{x}_{n,\ell}\mathbf{x}_{n,\ell}^H = \kappa \sum\nolimits_{q=1}^{N_t}[\boldsymbol{\Phi}]_{:,q}[\boldsymbol{\Phi}]_{:,q}^H = \frac{L}{N_t} \mathbf{I}_{N_t}.
	\end{equation}
	\begin{figure*}
		\vspace{-0.2cm}
		\begin{align}
			&	\!\!\!\!\psi_E(\boldsymbol{\theta}_0,\theta_0',\boldsymbol{\theta}_E,\theta_E',\!\boldsymbol{\tau}_E,\!\tau_E')\! \triangleq \! \sum_{k=1}^{K_s}|\tilde{z}^k_{1,e}\tilde{z}^k_{2}| = \!\sum_{k=1}^{K_s}\! \frac{\left|\tilde{\mathbf{a}}^H_{N_e}\!(\theta_E')\tilde{\mathbf{a}}_{N_e}\!(\theta^k_E)\right|}{N^p_sL}\!\!\left|\sum_{n \in \mathcal{N}_s^p}\!\!\!\mathbf{a}^H_{N_t}(\theta_0';\eta_s)\! \left(\sum_{\ell=1}^{L}\!\mathbf{x}_{n,\ell}\mathbf{x}^H_{n,\ell}\!\right)\!\mathbf{a}_{N_t}\!(\theta^k_0;\eta_s)\big[{\mathbf{d}_e}(\tau_E')\big]^*_n\big[{\mathbf{d}_e}(\tau^k_E)\big]_n\right| \label{eq_AF_intermediate_1} \\
			&\hspace{5cm}\stackrel{(a)}{=} \sum_{k=1}^{K_s} \underbrace{\left|\tilde{\mathbf{a}}^H_{N_e}(\theta_E')\tilde{\mathbf{a}}_{N_e}(\theta^k_E)\right|}_{\triangleq\psi_{EA}(\theta_E',\theta_E^k)} \times\underbrace{\left|\mathbf{a}^H_{N_t}(\theta_0';\eta_s) \mathbf{a}_{N_t}(\theta^k_0;\eta_s)\right|/N_t}_{\triangleq\psi_{BS}(\theta_0',\theta^k_0,\eta_s)} \times \underbrace{\left|\mathbf{d}_e^H(\tau_E')\mathbf{d}_e(\tau^k_E)\right|/N_s^p}_{\triangleq\psi_{ER}(\tau_E',\tau^k_E,\eta_f)},\label{eq_AF_intermediate}
		\end{align}
		\vspace{-0.2cm}
		\hrule
		\vspace{-0.4cm}
	\end{figure*}
	Now, using the definition of ${\mathbf{a}}_{N_t}(\theta;\eta_s)$, we obtain
	\begin{align}\label{eq_AF_derivation}
		\psi_{BS}(\theta_0',\theta^k_0,\eta_s) &\!=\frac{1}{N_t}\! \left|\sum\nolimits_{n=1}^{N_t}\!\!e^{j\pi(n-1)\eta_s\left(\sin(\theta_0')-\sin(\theta^k_0)\right)}\!\right|\!\!\\
		&\hspace{-2.1cm}=\frac{1}{N_t}\left|\frac{1-e^{j\pi N_t\eta_s\left(\sin(\theta_0')-\sin(\theta^k_0)\right)}}{1-e^{j\pi\eta_s\left(\sin(\theta_0')-\sin(\theta^k_0)\right)}}\right|= \left|\dfrac{\sin\left(\frac{\pi N_t}{2}\eta_s\Delta^k_{\theta_0',\theta_0}\right)}{N_t\sin\left(\frac{\pi}{2}\eta_s\Delta^k_{\theta_0',\theta_0}\right)}\right|, \label{eq_AF_final_form}
	\end{align}
	which yields~\eqref{eq_AMB_BS}. Similarly, we can derive~\eqref{eq_AMB_Eveangle} and~\eqref{eq_AMB_Everange}. Furthermore, the support of $\psi_{ER}(\tau_E',\tau^k_E,\eta_f)$ follows by recalling that Eve is concerned with targets only within the unambiguous range for which we have $|\tau_E'-\tau^k_E| < \frac{1}{\Delta f}$. Substituting for these expressions in~\eqref{eq_AF_intermediate}, completes the proof.
 \qed
\vspace{-0.15cm}		
\section{Proof of Lemma~\ref{lem_ambiguity_function}}\label{app_SP_AF_Peak}
	From~\eqref{eq_AF_final_form}, the (reparameterized) spatial MIMO-AF of a $N_t$-element USA with sampling factor $\eta_s$ can be written as
	\vspace{-0.1cm}
	\begin{equation}\label{eq_reparamter_AF}
		\!\!	\psi_{BS}(\Delta_s;\eta_s) =\frac{1}{N_t}
		\left|
		{
			\sin\left(\frac{\pi N_t}{2}\eta_s \Delta_s\right)
		}\Big/{
			\sin\left(\frac{\pi}{2}\eta_s \Delta_s\right)
		}
		\right|,
		\vspace{-0.1cm}
	\end{equation}
	where $\Delta_s$ is the spatial difference between a reference and target angle expressed in the normalized (i.e., direction-sine) domain. Now, since the numerator and denominator of $\psi_{BS}(\Delta_s;\eta_s)$ are bounded in magnitude by $1$, $\psi_{BS}(\Delta_s;\eta_s)$ attains its local maximum when the denominator approaches $0$. This happens when 
	$\sin\left(\frac{\pi}{2}\eta_s \Delta_s\right) = 0$, or equivalently when $\frac{\pi}{2}\eta_s \Delta_s = \pi n, \ n \in \mathbb{Z}.$ Rearranging this equation yields
	\vspace{-0.1cm}
	\begin{equation}
		\vspace{-0.1cm}
		\Delta_{s,n} = {2n}\big/{\eta_s}, \quad n \in \mathbb{Z}.
	\end{equation}
	Finally, since the spatial difference variable satisfies $\Delta_s \in \Omega_s$, only those indices satisfying $|2n/\eta_s|< 2$ are admissible. Hence the valid choice of indices for $n$ are those in~\eqref{eq_peak_AF_values}.
 \qed
\vspace{-0.15cm}
\section{Proof of Theorem~\ref{thm_ghost_target_location}}\label{app_SP_Ghost}

	Observe from Lemma~\ref{lem_ambiguity_function} that the AF $\psi_{BS}(\Delta;\eta_s)$ attains its peaks at $\Delta_{s,n}={2n}\big/{\eta_s}$,
	all with magnitude $\psi_{BS}(\Delta_n;\eta_s)=1$, equal to the main peak at $\Delta_s=0$. Then, since a ML estimator selects the peak locations of the AF, each $\Delta_{s,n}$ yields a candidate target direction. Noting that $\Delta_s=\sin(\theta'_0)-\sin(\theta_0)$, the peak condition becomes
	\begin{equation}
		\sin(\hat{\theta}_{0,n})-\sin(\theta_0)={2n}\big/{\eta_s}.
	\end{equation}
	Solving for $\hat{\theta}_{0,n}$ yields~\eqref{eq_ghost_target_location},
	and these correspond to the ghost target angles when $n\neq 0$. Finally, requiring $\sin(\hat{\theta}_{0,n})\in[-1,1)$ yields the admissible index set $\mathcal{N}_g(\theta_0)$ in~\eqref{eq_gh_location_ind}.  Thus, for $\eta_s> 1$, there exists at least one non-zero integer $n$, implying that multiple (ghost) DoD estimates are observed at Eve in addition to the true DoD. This completes the proof. \qed
\vspace{-0.2cm}	
	\section{Proof of Theorem~\ref{thm_localization_ambiguity}}\label{app_Position_AMB_Cond}
		Let a candidate target location be denoted by $\mathbf{p}_T$. Since the DoA at Eve, $\theta_E$, can naturally be uniquely estimated, any candidate target location estimated by Eve must lie along the ray originating at Eve and directed along the DoA observed. Hence, $\mathbf{p}_T$ can be expressed in its parametric form as
		\vspace{-0.1cm}
	\begin{equation}\label{eq_parametric_form}
		\mathbf{p}_T
		=
		\mathbf{p}_E + \beta \mathbf{u}_E,
		\qquad
		\beta > 0,
				\vspace{-0.1cm}
	\end{equation}
	where $\mathbf{u}_E = [\cos\theta_E,\sin\theta_E]^T$, and $\beta$ measures the distance between Eve and target. Recall from Theorem~\ref{thm_range_ambiguity} that Eve cannot distinguish the true delay $\tau_E$ from other ghost delays, as characterized in~\eqref{eq_delay_ghost_values}. 
	Furthermore, each such delay corresponds to a locus traced by a bistatic ellipse with focal points at the BS and Eve locations. Therefore, any candidate target location associated with $\tau_{E,m}$ must satisfy
	\begin{equation}
		\label{eq_ellipse_constraint}
		\|\mathbf{p}_T-\mathbf{p}_{BS}\|_2
		+
		\|\mathbf{p}_T-\mathbf{p}_E\|_2
		=
		c\tau_E + \frac{cm}{\eta_f \Delta f}.
	\end{equation}
	Next, recall  from Theorem~\ref{thm_ghost_target_location} that Eve cannot distinguish the true DoD $\theta_0$ from other ghost DoDs given by~\eqref{eq_ghost_target_location}. 
	Now for an arbitrary location $\mathbf{y}$, the corresponding DoD at the BS satisfies
	\begin{equation}
		\sin(\theta(\mathbf{y}))
		=
		\frac{
			\mathbf{e}_2^T(\mathbf{y}-\mathbf{p}_{BS})
		}
		{
			\|\mathbf{y}-\mathbf{p}_{BS}\|_2
		}.
	\end{equation}
	Hence, for a candidate target location to be consistent with one of the ghost DoDs characterized in Theorem~\ref{thm_ghost_target_location}, the following relationship has to hold: 
	\begin{equation}
		\label{eq_dod_constraint}
		\frac{
			\mathbf{e}_2^T(\mathbf{p}_T-\mathbf{p}_{BS})
		}
		{
			\|\mathbf{p}_T-\mathbf{p}_{BS}\|_2
		}
		=
		\sin(\theta_0)+\frac{2n}{\eta_s}.
	\end{equation}
	Now, it follows from~\eqref{eq_ellipse_constraint} that
	\begin{equation}
		\|\mathbf{p}_T-\mathbf{p}_{BS}\|_2
		=
		c\tau_E+\frac{cm}{\eta_f\Delta f}
		-\|\mathbf{p}_T-\mathbf{p}_E\|_2.
	\end{equation}
	Using~\eqref{eq_parametric_form} in the above yields
	\begin{equation}\label{eq_final_difference_locations}
		\|\mathbf{p}_T-\mathbf{p}_{BS}\|_2
		=
		c\tau_E+\frac{cm}{\eta_f\Delta f}
		-\beta\|\mathbf{u}_E\|_2
		=
		c\tau_E+\frac{cm}{\eta_f\Delta f}
		-\beta,
	\end{equation}
	where the last equality follows because $\|\mathbf{u}_E\|_2=1$. 
	Finally, using $\mathbf{p}_T-\mathbf{p}_{BS}= \mathbf{p}_E+\beta\mathbf{u}_E-\mathbf{p}_{BS}$,
	together with~\eqref{eq_final_difference_locations}
	in~\eqref{eq_dod_constraint},
	we obtain~\eqref{eq_localization_amb_eqn}.
	
	In summary, if there exist integers $n \neq 0, m\neq 0$ and a scalar $\beta>0$ satisfying \eqref{eq_localization_amb_eqn}, then there exists at least one ghost candidate location
	$\mathbf{p}_g \neq \mathbf{p}_T$ given by
			\vspace{-0.1cm}
	\begin{equation}
		\mathbf{p}_g
		=
		\mathbf{p}_E+\beta\mathbf{u}_E,
				\vspace{-0.1cm}
	\end{equation}
	which simultaneously satisfies:
	\begin{enumerate}
		\item the observed DoA $\theta_E$,
		\item a valid ghost delay $\tau_{E,m}$, for some $m \in [|\mathcal{M}_g(\tau_E)|]$, and
		\item a valid ghost DoD $\theta_{0,n}$, for some $n \in [|\mathcal{N}_g(\theta_0)|]$.
	\end{enumerate}
	Consequently, the parameter tuple associated with
	$\mathbf{p}_g$
	is indistinguishable from that of the true target based on Eve's observations.
	Therefore, Eve cannot uniquely determine the target location, thereby preserving sensing privacy. \qed
		\section{Proof of Proposition~\ref{prop_existance_of_position_error}}\label{app_Position_Amb_satisfy}
		
		Define the ambiguity offsets
	\vspace{-0.1cm}
	\begin{equation}
		C \triangleq C(n)
		=
		\sin(\theta_0)+\frac{2n}{\eta_s},
		\
		D \triangleq D(m) =
		c\tau_E+\frac{cm}{\eta_f\Delta f},
				\vspace{-0.1cm}
	\end{equation}
	and the geometry-dependent quantities
	\vspace{-0.1cm}
	\begin{equation}
		A = \mathbf{e}_2^T(\mathbf{p}_E-\mathbf{p}_{BS}), \quad
		B = \mathbf{e}_2^T\mathbf{u}_E,
				\vspace{-0.1cm}
	\end{equation}
	where all symbols have the same meaning as in Theorem~\ref{thm_localization_ambiguity}. Then, using these definitions, the positioning ambiguity condition in Theorem~\ref{thm_localization_ambiguity} can be written as
	\vspace{-0.1cm}
	\begin{equation}
		A+\beta B
		=
		C(D-\beta).
		\vspace{-0.1cm}
	\end{equation}
	Rearranging the above yields	$\beta = \big(CD-A)\Big/\big(B+C\big)$.
	Thus, the existence of a feasible ghost target location requires
	\vspace{-0.1cm}
	\begin{equation}\label{eq:feasibility}
		\beta>0
		\Longleftrightarrow 
		(CD-A)(B+C)>0.
		\vspace{-0.1cm}
	\end{equation}
	Next, observe that the ambiguity spacings in the DoD and delay domains are given by $
	\Delta C
	=
	2/\eta_s$, and
	$\Delta D
	=
	c/\eta_f\Delta f$. 
	Hence, as $\eta_s,\eta_f$ become very large, the ambiguity offsets ${C(n)}$ and ${D(m)}$ become arbitrarily dense, and the variables $C$ and $D$ may be treated as continuous quantities. To this end, we now choose a non-degenerate ghost DoD ray, i.e., $C,n\neq 0$, so that $B+C\neq 0$, and then analyse the feasibility of~\eqref{eq:feasibility}. 
	
	\textit{Case 1: $C>0$.}	If $B+C>0$, then \eqref{eq:feasibility} reduces to
	\vspace{-0.1cm}
	\begin{equation}\label{eq_D_condition}
		CD-A>0
		\Longleftrightarrow
		D>A/C.
				\vspace{-0.1cm}
	\end{equation}
	Since $D$ can take arbitrarily fine values as $\eta_f\rightarrow\infty$, a $D$ satisfying~\eqref{eq_D_condition} always exists for an appropriate value of $m$.
	
	On the other hand, if $B+C<0$, then \eqref{eq:feasibility} reduces to
	\vspace{-0.1cm}
	\begin{equation}
		\vspace{-0.1cm}
		CD-A<0
		\Longleftrightarrow
		D<A/C,
	\end{equation}
	and a feasible value of $D$ can again be selected.
	
	\textit{Case 2: $C<0$.} If $B+C>0$, then \eqref{eq:feasibility} reduces to
	\vspace{-0.1cm}
	\begin{equation}
		\vspace{-0.1cm}
		CD-A>0
		\Longleftrightarrow
		D<A/C,
				\vspace{-0.1cm}
	\end{equation}
	and a feasible choice of $D$ satisfying the above exists. Similarly, when $B+C<0$, then \eqref{eq:feasibility} reduces to
	\vspace{-0.1cm}
	\begin{equation}
		\vspace{-0.1cm}
		CD-A<0
		\Longleftrightarrow
		D>A/C,
				\vspace{-0.1cm}
	\end{equation}
	which again admits a feasible choice of $D$. In hindsight, the above analysis reveals that ambiguity in the spatial and frequency domains can compensate for one another, so that a suitable choice of $(C,D)$ satisfies \eqref{eq:feasibility}. Since the ambiguity sets $\{C(n)\}$ and $\{D(m)\}$ become arbitrarily dense as $\eta_s,\eta_f\to\infty$, there exists a non-trivial set of integers $n\neq 0,m\neq 0$ satisfying the condition in Theorem~\ref{thm_localization_ambiguity}. \qed
			\vspace{-0.1cm}
	\section{Proof of Lemma~\ref{lem_AF_Alice_expression}}\label{app_Alice_AF}
	
	Since Alice has the knowledge of both data and pilot samples, she performs sensing using all OFDM SCs. Then, the impulse response of MF at Alice is given by
	\vspace{-0.1cm}
	\begin{equation}\label{eq_MF_IR_Alice}
		\vspace{-0.1cm}
		\mathbf{G}_{a,MF} = \tilde{\mathbf{a}}_{N_r}(\theta_A')\Big\{\left(\mathbf{a}^H_{N_t}(\theta_0';\eta_s)\mathbf{X}\right) \odot \left(\mathbf{d}_a(\tau_A')\otimes \mathbf{1}_L\right)^T\Big\}.
				\vspace{-0.1cm}
	\end{equation}
	Then, similar to the derivation of~\eqref{eq_AF_intermediate_1}, we can obtain the range-angle MIMO-AF as in~\eqref{eq_AF_intermediate_1_Alice}. 
	\begin{figure*}
		\vspace{-0.0cm}  
		\begin{equation}\label{eq_AF_intermediate_1_Alice}
			\psi_A(\boldsymbol{\theta}_0,\theta_0',\boldsymbol{\theta}_A,\theta_A',\boldsymbol{\tau}_A,\tau_A') = \sum_{k=1}^{K_s} \frac{\left|\tilde{\mathbf{a}}^H_{N_r}(\theta_A')\tilde{\mathbf{a}}_{N_r}(\theta^k_A)\right|}{N_sL}\left|\sum_{n \in \mathcal{N}_s}\mathbf{a}^H_{N_t}(\theta_0';\eta_s) \left(\sum_{\ell=1}^{L}\mathbf{x}_{n,\ell}\mathbf{x}^H_{n,\ell}\right)\mathbf{a}_{N_t}(\theta^k_0;\eta_s)\big[{\mathbf{d}_a}(\tau_A')\big]^*_n\big[{\mathbf{d}_a}(\tau^k_A)\big]_n\right|. 
		\end{equation}
		\hrule
		\vspace{-0.2cm}
	\end{figure*}
	In~\eqref{eq_AF_intermediate_1_Alice}, to simplify the second term containing the sum over all SCs, we note the following: 
	\begin{enumerate}[leftmargin=*]
		\item When $n \in \mathcal{N}_s^p$, from~\eqref{eq_pilot_SC_signal_OP}, we get $\sum_{\ell=1}^{L}\mathbf{x}_{n,\ell}\mathbf{x}^H_{n,\ell} = \frac{L}{N_t}\mathbf{I}_{N_t}$. 
		\item When $n \in \mathcal{N}_s^d$, since the data symbols are random, by the law of large numbers,\footnote{Note that the random sidelobes that may arise from data-payload-based sensing get suppressed due to the large CPI.} $\sum_{\ell=1}^{L}\!\mathbf{x}_{n,\ell}\mathbf{x}^H_{n,\ell} \! \rightarrow \! L\boldsymbol{\Sigma}_{X,n} \! = \frac{L}{N_t}\mathbf{I}_{N_t}$.
	\end{enumerate}
	Furthermore, under monostatic setting, we have $\theta^k_0=\theta^k_A, \forall k $. Thus, it is sufficient to estimate the DoA, which will provide complete information about the target in the SD. Then, the overall range-angle MIMO-AF at Alice becomes
	\vspace{-0.1cm}
	\begin{multline}
		\vspace{-0.1cm}
		\psi_A(\boldsymbol{\theta}_0,\theta_0',\boldsymbol{\theta}_A,\theta_A',\boldsymbol{\tau}_A,\tau_A') =  \\ \sum\nolimits_{k=1}^{K_s} \left|\tilde{\mathbf{a}}^H_{N_r}(\theta_A')\tilde{\mathbf{a}}_{N_r}(\theta^k_A)\right| \times {\left|\mathbf{d}_a^H(\tau_A')\mathbf{d}_a(\tau^k_A)\right|}\big/{N_s}.
	\end{multline}
	Simplifying the above similar to~\eqref{eq_AF_final_form} completes the proof. \qed
	
		\section{Proof of Theorem~\ref{prop_legitimate_performance}}\label{app_legitimate_performance}
	 We prove the theorem separately for sensing at Alice and for communication with a UE.
	
	\underline{Sensing:}
	It is readily observed from Lemma~\ref{lem_AF_Alice_expression} that the range-angle AF at Alice, $\psi_A(\boldsymbol{\theta}_0,\theta_0',\boldsymbol{\theta}_A,\theta_A',\boldsymbol{\tau}_A,\tau_A')$ scales as~\eqref{eq_eve_AF_seperable_form} with $N_e=N_r$ and $\eta_s=\eta_f=1$. Hence the system performs similar to a Nyquist-sampled system in the spatial-frequency domain. Furthermore, the AF at Alice for a given $k$th target obeys:
	\begin{equation}
		\left|
		\frac{\sin\!\left(\frac{\pi N_r}{2}\Delta^k_\theta\right)}
		{N_r \sin\!\left(\frac{\pi}{2}\Delta^k_\theta\right)}
		\right|
		\cdot
		\left|
		\frac{\sin\!\left(\pi N_s \Delta f \Delta^k_\tau\right)}
		{N_s \sin\!\left(\pi \Delta f \Delta^k_\tau\right)}
		\right| \stackrel{(a)}{\rightarrow}
		\mathcal{F}_{N_r}(\Delta^k_\theta)\,\mathcal{F}_{N_s}(\Delta^k_\tau),
	\end{equation}
	where $\mathcal{F}_{x}(\cdot)$ is the Fejér Kernel, and in $(a)$, $N_r,N_s \rightarrow \infty$.
	Since a Fejér kernel has a unique maximum at $0$ and vanishes elsewhere, the $k$th sum-component of $\psi_A(\boldsymbol{\theta}_0,\theta_0',\boldsymbol{\theta}_A,\theta_A',\boldsymbol{\tau}_A,\tau_A')$  has a unique global maximum at $(0,0)$ over its domain $\Omega_s \times \Omega_f$. This implies unique identifiability of target parameters by Alice, independent of $\eta_s$, $\eta_f$.\\
	\indent \underline{Communication Rate:} 
	On a data SC-$n \in \mathcal{N}_s^d$, the signal received in an OFDM symbol can be written as
	\vspace{-0.1cm}
	\begin{equation}
		\mathbf{y}_{c,n} = \mathbf{H}_{c,n} \mathbf{x}_{c,n} + \mathbf{n}_{c,n} \stackrel{(b)}{=} \hat{\mathbf{H}}_{c,n}\mathbf{x}_{c,n} + \tilde{\mathbf{H}}_{c,n}\mathbf{x}_{c,n} +  \mathbf{n}_{c,n},
	\end{equation}
	where in $(b)$, $\tilde{\mathbf{H}}_{c,n}$ is the channel estimation noise. Then, using the worst-case noise theorem~\cite{Hassibi_TIT_2003}, we treat $\tilde{\mathbf{H}}_{c,n}$ as Gaussian self-interference,  and bound the achievable rate in~\eqref{eq_througput} as
	\begin{equation}
		R \geq \frac{1}{N_s} \sum\nolimits_{n \in \mathcal{N}_s^d}
		\log_2 \det\!\left(
		\mathbf{I}_{N_c} + \frac{P}{N_tN_s} \boldsymbol{\Sigma}_{n,IN}^{-1}\hat{\mathbf{H}}_{c,n}\hat{\mathbf{H}}_{c,n}^H
		\right),
	\end{equation}
	where $\boldsymbol{\Sigma}_{n,IN} = \mathbb{E}[(\tilde{\mathbf{H}}_{c,n}\mathbf{x}_{c,n} +  \mathbf{n}_{c,n})(\tilde{\mathbf{H}}_{c,n}\mathbf{x}_{c,n} +  \mathbf{n}_{c,n})^H]$ is the interference-plus-noise covariance matrix. Now, exploiting the properties of all the variables, we can show that
	\vspace{-0.1cm}
	\begin{equation}
		\boldsymbol{\Sigma}_{n,IN} = \frac{1}{N_t} \sum_{i=1}^{N_t}(\mathbf{e}_i^T \otimes \mathbf{I}_{N_r})\mathbf{E}_n(\mathbf{e}_i^T \otimes \mathbf{I}_{N_r})^H + \sigma_w^2\mathbf{I}_{N_r},
		\vspace{-0.1cm}
	\end{equation}
	where $\mathbf{E}_n$ is as given in Lemma~\ref{lemma_CSI_error_variance}, and $\mathbf{e}_i \in \mathbb{C}^{N_t}$ is the $i$th canonical basis vector. Let $\mathbf{G}_{c,n} \triangleq \boldsymbol{\Sigma}_{n,IN}^{-1/2}\hat{\mathbf{H}}_{c,n}$ and $\{\mu_{k,n}\}_{k}$ be its singular values. Then, by matrix theory properties,
	\vspace{-0.1cm}
	\begin{equation}
		R \geq \frac{1}{N_s} \sum\nolimits_{n \in \mathcal{N}_s^d}
		\sum\nolimits_{k=1}^{\mathrm{rank}(\mathbf{G}_n)}
		\log_2\!\left(1 + \frac{P}{N_tN_s}\mu_{k,n}\right).
	\end{equation}
	Note that $\mathrm{rank}(\mathbf{G}_n) = \mathrm{rank}(\hat{\mathbf{H}}_{c,n})$ since $\boldsymbol{\Sigma}_{n,IN}$ has full-rank. Finally, observing that $|\mathcal{N}_s^d| = N_s - N_s^p$ and $\eta_f = N_s/N_s^p$, the achievable rate thus scales order-wise as given in~\eqref{eq_Rate_Comm_UE}. \qed

\end{appendices}
\bibliographystyle{IEEEtran}
\bibliography{IEEEabrv,References}	
\end{document}